\documentclass[11pt]{article}
\usepackage{amsmath}
\usepackage{amssymb}
\usepackage{graphicx}
\usepackage{epsfig}
\usepackage{algorithm}

\usepackage{epstopdf}
\usepackage{subfigure}
\def\dis{\displaystyle}
\newenvironment{proof}[1][Proof]{\noindent\textbf{#1.} }{\ \rule{0.5em}{0.5em}}

\begin{document}

\newtheorem{Thm}{Theorem}[section]
\newtheorem{Cor}[Thm]{Corollary}
\newtheorem{Lem}[Thm]{Lemma}
\newtheorem{Prop}[Thm]{Proposition}
\newtheorem{Def}[Thm]{Definition}
\newtheorem{rem}[Thm]{Remark}

\newtheorem{Assump}[Thm]{Assumption}
 \baselineskip 0.20in

\title{\textbf{Heavy-Ball-Based Hard Thresholding Algorithms for  Sparse Signal Recovery}\thanks{The work was founded by the National Natural Science Foundation of China
(NSFC\#12071307 and 11771255), Young Innovation Teams of Shandong
Province (\#2019KJI013), and  Domestic and Oversea  Visiting Program for  the Middle-aged and Young Key Teachers of  Shandong University of Technology.}}

\author{Zhong-Feng Sun\thanks{School of Mathematics and Statistics, Shandong
University of Technology, Zibo, Shandong,  China. (e-mail: {\tt
zfsun@sdut.edu.cn}).}~, Jin-Chuan Zhou\thanks{School of Mathematics and Statistics, Shandong
University of Technology, Zibo, Shandong,  China.  (e-mail: {\tt
jinchuanzhou@163.com}).}~,  Yun-Bin Zhao\thanks{Corresponding author. Shenzhen Research Institute of Big Data, Chinese University of Hong Kong, Shenzhen,  Guangdong, China.  (e-mail: {\tt
yunbinzhao@cuhk.edu.cn}).}~, and   Nan Meng\thanks{School of Mathematics, University of Birmingham, Edgbaston, Birmingham B15 2TT, United
Kingdom.
  (e-mail: {\tt nxm563@bham.ac.uk}).}}

\date{ }

\maketitle

\noindent
 \textbf{Abstract.} The hard thresholding technique plays a vital role in the  development of algorithms for sparse signal recovery. By merging this technique and heavy-ball acceleration method which is a multi-step  extension
of the traditional gradient descent method, we propose the so-called heavy-ball-based hard thresholding (HBHT) and heavy-ball-based hard thresholding pursuit (HBHTP) algorithms for signal recovery.
It turns out that the HBHT and HBHTP can successfully recover a $k$-sparse signal if the restricted isometry constant of the measurement matrix  satisfies $\delta_{3k}<0.618 $ and $\delta_{3k}<0.577,$ respectively.  The  guaranteed success of HBHT and HBHTP  is also shown under the conditions $\delta_{2k}<0.356$ and $\delta_{2k}<0.377,$ respectively. Moreover, the finite convergence and stability of the two algorithms are also established in this paper.  Simulations on random problem instances are performed to compare the performance of the proposed algorithms and several existing ones. Empirical results indicate that the HBHTP performs very comparably to a few existing algorithms and it takes less average time to achieve the signal recovery than these existing methods. \\



\section{Introduction}

In compressed sensing scenarios, one needs to recover a sparse signal $x\in \mathbb{R}^n$ from linear measurements $y:= Ax+\nu$, where $\nu \in \mathbb{R}^m  $ are measurement errors and $ A$ is a known $ m\times n$ measurement matrix with $m \ll n.$  When $ \nu=0$,  the measurements $y$ are accurate. To recover the signal $x$ in such an environment, one may consider the optimization model
 \begin{equation}\label{main-optimal-eq}
\underset{z}{\min}\{ {\left\lVert y-Az\right\rVert}_2^2: \left\lVert z\right\rVert_0\leq k\},
\end{equation}
where $ k$ (a given integer number) is an estimate of the sparsity level of $x,$  and $ \left\lVert z\right\rVert_0$ denotes the number of nonzero entries
of $z\in \mathbb{R}^{n}.$  In this paper, a vector $z$ is said to be $k$-sparse if $\left\lVert z\right\rVert_0\leq k$. It is well known that when $x$ is $k$-sparse and $A$ satisfies certain assumptions,  $x$ will be the unique $k$-sparse solution to the  problem  (\ref{main-optimal-eq}) (see, e.g., \cite{elad2010sparse,Foucart-2013,zhao2018sparse}). Thus the recovery of $x$ often amounts to solving the problem (\ref{main-optimal-eq}), and the algorithms for such a prpoblem are usually called compressed sensing algorithms or, in more general, sparse optimization algorithms.

Let us first briefly review the thresholding algorithms for sparse signal recovery. The thresholding technique was introduced by Donoho and  Johnstone \cite{donoho1994ideal}. At present,  there are three main classes of thresholding algorithms: hard thresholding \cite{blanchard2015cgiht, blumensath2012accelerated, blumensath2008iterative, blumensath2009iterative, blumensath2010normalized, cevher2011accelerated, Foucart-2011, khanna2018iht, kyrillidis2014matrix}, soft thresholding \cite{daubechies2004iterative,donoho1995noising,elad2006simple}, and optimal thresholding \cite{meng2021partial,Zhao-2020,Zhao-2020-Analysis}. A huge amount of  work has been carried out for the class of hard thresholding algorithms. For instance, the iterative hard thresholding (IHT)  was early studied in  \cite{blumensath2008iterative}, which is a combination of the  gradient descent and hard thresholding technique. The IHT admits a few modifications including  the gradient descent with sparsification (GDS) \cite{garg2009gradient} and the normalized iterative hard thresholding  with a fixed or adaptive steplength \cite{blumensath2012accelerated,blumensath2010normalized}. In addition, combining IHT \cite{blumensath2008iterative,blumensath2010normalized,garg2009gradient} and  orthogonal projection immediately leads to the hard thresholding pursuit (HTP)  in \cite{Foucart-2011}. The thresholding methods combined with Nesterov's acceleration technique were also studied (e.g.,  \cite{cevher2011accelerated,khanna2018iht,kyrillidis2014matrix}).
Recently, it was pointed out in \cite{Zhao-2020} that performing hard thresholding is usually independent of the reduction of the residual $\|y-Az\|_2^2$ and thus the so-called optimal $k$-thresholding operator was proposed in \cite{Zhao-2020,Zhao-2020-Analysis}. Nevertheless, the optimal $k$-thresholding  algorithm need to solve a quadratic convex optimization problem at every iteration which requires more computational time than the traditional hard thresholding  methods.

The aim of this paper is to use the heavy ball method to accelerate the hard thresholding algorithms without increasing its computational complexity.  Recall that the search direction at the iterate $x^p$ in IHT and HTP is given by $A^ T(y-Ax^p)$, which is the negative gradient of the residual at $ x^p.$ With the aid of  the momentum term $x^p-x^{p-1}$, the search direction can be  modified to   \begin{equation} \label{DIRECT}
d^{p}=\alpha A^ T(y-Ax^p)+\beta (x^p-x^{p-1})
\end{equation} with two parameters $\alpha>0$ and $  \beta \geq 0. $ This is the two-step heavy-ball method proposed for optimization problems by Polyak \cite{polyak1964some}. It
has been shown that a fast local convergence of this method for optimization problems can be achieved provided that the parameters are properly chosen, and that the method can  work even when the Hessian matrix of the objective function is ill-posed  (see, e.g., Chapter 3 in \cite{polyak1987introduction}). The global convergence of the heavy ball method has also been discussed in the literature \cite{ghadimi2015global,gurbuzbalaban2017on,lessard2016analysis,xin2020distributed}. This method is widely used in such fields  as distributed optimization \cite{gurbuzbalaban2017on,xin2020distributed}, variational inequality \cite{Huang2021A},  wireless network \cite{liu2016heavy}, nonconvex optimization \cite{ochs2014ipiano,Sun-2019-Heavy},  deep neural network \cite{tao2021the},  and image restoration \cite{wang2014scaled}. For instance, it was found in \cite{Sun-2019-Heavy} that the
heavy ball momentum plays an important role in driving the iterates away from the saddle points of nonconvex optimization problems;
It was also used in \cite{wang2014scaled} to accelerate the Richardson-Lucy algorithm in image deconvolution without causing a remarkable increase of iteration complexity;
Xin and Khan \cite{xin2020distributed} observed that the distributed heavy ball method achieves a global $R-$linear rate for distributed optimization, and the momentum term can dramatically improves the convergence of the algorithm for ill-conditioned objective functions. These and other applications indicate that the heavy-ball method does admit certain advantage in enhancing the efficiency of an iterative method for optimization problems and  may outperform the extra-point method and Nesterov acceleration method.

Motivated by the numerical advantage of heavy ball acceleration technique,  we propose the heavy-ball-based hard thresholding (HBHT) and heavy-ball-based hard thresholding pursuit (HBHTP) algorithms for the recovery problem (\ref{main-optimal-eq}).
The  guaranteed performance of the two algorithms are shown under the assumption of restricted isometry property (RIP), which was
originally introduced by Cand\`{e}s and Tao \cite{Candes-2005-Decoding} and has now become a standard tool for the analysis of various compressed sensing algorithms.
It is well known that the success of IHT for $k$-sparse signal recovery can be guaranteed under the RIP condition $\delta_{3k}<  (\sqrt{5}-1)/2 \approx 0.618$ (see \cite{Zhao-2020-Improved}) and  that of HTP can be guaranteed  under  $\delta_{3k}<1/\sqrt{3} \approx 0.577$ (see  \cite{Foucart-2011}).
Under the same condition and a proper choice of algorithmic parameters, we establish the guaranteed-performance results for the two algorithms HBHT and HBHTP. Roughly speaking, we show that the HBHT is convergent under the  condition $\delta_{3k}<(\sqrt{5}-1)/2$  and that  HBHTP is convergent under the condition $\delta_{3k}<1/\sqrt{3}.$  By using an analysis method in \cite{Zhao-2020-Analysis}, we further prove that the condition for theoretical performance of the two algorithms can be established in term of  $\delta_{2k}$ as well. Specifically,   the guaranteed success of HBHT and HBHTP can be ensured if $\delta_{2k}<(\sqrt{5}-1)/(2\sqrt{3}) \approx 0.356$  and $\delta_{2k}<1/\sqrt{7} \approx0.377$, respectively.  Moreover, the finite convergence and recovery stability of the two methods are also shown in this paper.

  A large amount of experiments on random problem instances of sparse signal recovery are performed to investigate the success rate and phase transition features of the proposed algorithms. We also compare the performances of the proposed algorithms and several existing ones such as orthogonal matching pursuit (OMP) \cite{davis1994adaptive,tropp2007signal}, compressive sampling matching pursuit (CoSaMP) \cite{needell2009cosamp},  subspace pursuit (SP) \cite{dai2009subspace}, IHT and HTP. The empirical results show that incorporating heavy ball technique into IHT and HTP  does remarkably improve the performance of IHT and HTP, respectively. The HBHTP  not only admits robust signal recovery ability in both noisy and noiseless scenarios, but also takes relatively less average computational time to achieve the recovery success compared to a few existing methods.

The paper is structured as follows. In Section \ref{sec-prelim}, we described the HBHT and HBHTP algorithms and list some notations and useful inequalities. The theoretical analysis of the proposed algorithms is conducted in Sections \ref{convergence} and \ref{stability}. Numerical results are given in Section \ref{simulation}, and conclusions are drawn in the last section.

\section{Preliminary and algorithms}\label{sec-prelim}
\subsection{Notation}
Denote by $N:=\{1,2,\ldots,n\}.$  For a subset $\Omega\subseteq N,$  let $\overline{\Omega}:=N\setminus\Omega$ and $|\Omega|$ denote the complement set  and the cardinality of $\Omega$, respectively. Given a vector $z\in\mathbb{R}^{n}$, the index set $supp(z):=\{i\in N:z_i\neq 0\}$ denotes the support of $z$, and $z_\Omega\in \mathbb{R}^n$ is the vector with entries
$$(z_\Omega)_i=\left\{\begin{array}{cc}
z_i, &i\in \Omega,\\
0, &i\notin\Omega.\\
\end{array}\right.$$
Let $\mathcal{L}_k(z)$ be the index set of the  $k$ largest absolute entries of $z,$ and let $\mathcal{H}_k(\cdot)$ be the hard thresholding operator which retains the  $k$ largest magnitudes and zeroing out other entries of a vector.  The $k$-sparse vector $\mathcal{H}_k(z)$ is the best $k$-term approximation of $ z\in \mathbb{R}^n. $   Denote by $\sigma_k(z)_q,$ where $ q>0$ is an integer number, the residual of the best $k$-term approximation of $z$, i.e.,
$$ \sigma_k(z)_q= \min_{u}\{ \|z-u\|_q: \| u\|_0 \leq k\}. $$

\subsection{Basic inequalities}

We first recall the restricted isometry constant (property) of a given measurement matrix.

\begin{Def} \emph{\cite{Candes-2005-Decoding}} \label{def-RIC}
Let $A\in \mathbb{R}^{m\times n}$ with $m< n$ be a matrix. The restricted isometry constant (RIC) of order $k$, denoted
$\delta_k,$ is the smallest number  $\delta\geq 0$ such that
\begin{equation}\label{def-RIC-1}
(1-\delta){\left\lVert u \right\rVert}^2_2\leq  {\left\lVert Au \right\rVert}^2_2\leq (1+\delta){\left\lVert u \right\rVert}^2_2
\end{equation}
for all $k$-sparse vectors $u\in \mathbb{R}^n$(i.e., $\left\lVert u\right\rVert_0 \leq  k$). If $\delta_k<1$, then $A$ is
said to satisfy the restricted isometry property (RIP) of order
$k$.
\end{Def}

From the definition above, one can see that $ \delta_t\leq \delta_s $ for any integer number $ t \leq s. $ The following properties of RIC have been frequently used in the analysis of compressed sensing algorithms
.
\begin{Lem} \emph{\cite{Foucart-2011,Zhao-2020}} \label{lem-1}
Let $v \in \mathbb{R}^n$ be a vector, $t \in N$ be a positive integer number and $W\subseteq N$ be an index set. 
\begin{itemize}
\item [(i)]  If  $ |W\cup supp(v)|\leq t$, then
$
 \left\lVert \left ( (I-A^ T A)v\right)_W\right\rVert_2\leq \delta_t {\left\lVert v \right\rVert}_2.
$

\item[(ii)] If  $ |W|\leq t$, then
$
\left\lVert \left ( A^ Tv\right)_W\right\rVert_2\leq \sqrt{1+\delta_t} {\left\lVert v \right\rVert}_2.
$
\end{itemize}
\end{Lem}

The next lemma is taken directly from  \cite{Zhao-2020-Improved}, and can also be implied  from the result in \cite{shen2018a}.

\begin{Lem} \emph{\cite{Zhao-2020-Improved}} \label{lem-optimal}
For any vector $z \in \mathbb{R}^n$ and for any $k$-sparse
vector $x\in \mathbb{R}^n$, one has
$$
\left\lVert x-\mathcal{H}_k(z)\right\rVert_2 \leq
\eta \left\lVert \left (x - z\right)_{W\cup W^*}\right\rVert_2,
$$
where $\eta=(\sqrt{5}+1)/2,$ $W= supp(x)$ and $W^*= supp(\mathcal{H}_k(z))$.
\end{Lem}

\subsection{Algorithm}

We now describe the algorithms in this paper. The basic idea is to use the search direction $d^p$ given by (\ref{DIRECT}), resulted from the heavy-ball acceleration method, to generate a point $ u^p= x^p+d^p$ where $ x^p$ denotes the current iterate of the algorithm. Then performing a hard thresholding on $ u^p$ to produce the next iterate.   This idea leads to Algorithm 1 for the recovery problem (\ref{main-optimal-eq}).

\begin{algorithm}\caption{Heavy-Ball-Based Hard Thresholding (HBHT)}
Input $(A,y,k)$ and two parameters $\alpha> 0$ and $ \beta\geq 0 $ and two initial points $x^0 $ and $ x^1$.
\begin{itemize}
\item[ S1.] At $x^p$, set
\begin{equation}\label{algorithm-1-1}
u^p=x^p+\alpha A^ T(y-Ax^p)+\beta (x^p-x^{p-1}).
\end{equation}
[Note: The term $\beta (x^p-x^{p-1})$ is called the momentum term in heavy ball method.]

\item[S2.] Let $$x^{p+1}=\mathcal{H}_k(u^p).$$
 Repeat the above steps until a certain stopping criterion is satisfied.

\end{itemize}
\end{algorithm}

We may treat the point $ \mathcal{H}_k(u^p)$ in HBHT as an intermediate point and perform a pursuit step (i.e., orthogonal projection) to generate the next iterate $ x^{p+1}.$  This leads to Algorithm 2 called HBHTP.

\begin{algorithm}
\caption{Heavy-Ball-Based Hard Thresholding Pursuit (HBHTP).}
Input $(A,y,k)$ and two parameters $\alpha> 0$ and $ \beta\geq 0 $ and two initial points $x^0 $ and $ x^1.$
\begin{itemize}
\item[ S1.] At $x^p$, set
$$
u^p=x^p+\alpha A^ T(y-Ax^p)+\beta (x^p-x^{p-1}).
$$

\item[S2.] Let $S^{p+1}=\mathcal{L}_k(u^p)$, and
 \begin{equation}\label{algorithm-2-3}
x^{p+1}=\arg \min_{z\in \mathbb{R}^n} \{ {\left\lVert y-Az\right\rVert}_2^2: supp(z)\subseteq S^{p+1}\}.
\end{equation}
Repeat the above steps until a certain stopping criterion is satisfied.
\end{itemize}
\end{algorithm}

Clearly, the HBHT and HBHTP reduce to IHT and HTP, respectively, when $\alpha=1$ and $ \beta=0.$
The initial point $ x^0$ and $ x^1$ can be any vectors. The simplest choice is $ x^0= x^1=0.$  To stop the algorithms, one can set the maximum number of iterations or use other stopping criteria such as $ \|y-Ax^p\|_2 \leq \varepsilon ,$ where $ \varepsilon >0$ is a small tolerance. For instance, if the measurements $y$  are accurate enough, then measurement error $ \|\nu \|_2= \|y-Ax\|_2$ would be very small. In such a case, it makes sense to use the stopping criterion $ \|y-Ax^p\|_2 \leq \varepsilon .$  It is also worth mentioning that we treat the parameters $\alpha$ and $\beta $ as fixed input data of the proposed algorithms for simplicity and convenience of analysis  throughout the paper. However, it should be pointed out that these parameters can be updated from step to step in order to get a better performance of the algorithms from both theoretical and practical viewpoints. This might be an interesting future work.

\section{Analysis of HBHT and HBHTP} \label{convergence}

In this section, we analyze the performance of HBHT and HBHTP  under the RIP of order $3k$ and $2k$, respectively. We also discussed the finite convergence of HBHTP under some conditions. Since  the heavy-ball method is a two-step method in the sense that the next iterate $ x^{p+1}$ is generated based on the previous two iterates $ x^{p}$ and $ x^{p-1}$, the analysis of HBHT and HBHTP  is remarkably different from the traditional IHT and HTP. To this need, we first establish the following useful lemma.

\begin{Lem}\label{lem-two-level-geometric}
Suppose that the nonnegative sequence $\{a^p\}\subseteq\mathbb{R} ~ (p= 0, 1, \dots)$ satisfies
\begin{equation}\label{lem-tlgeom-1}
a^{p+1}\leq b_1a^{p}+b_2a^{p-1}+b_3,\ \ p\geq 1,
\end{equation}
where  $b_1,b_2,b_3\geq 0 $ and $b_1+b_2<1$. Then
\begin{equation}\label{lem-tlgeom-2}
a^{p}\leq\theta^{p-1}\left[a^{1}+(\theta-b_1)a^{0})\right]+\frac{b_3}{1-\theta},
\end{equation}
with
\begin{equation*}
0\leq \theta:=\frac{b_1+\sqrt{b_1^2+4b_2}}{2} <1.
\end{equation*}

\end{Lem}
\begin{proof}
Denote by $q_1:=\frac{-b_1+\sqrt{b_1^2+4b_2}}{2}$. Note that   $b_1, b_2\geq 0$ and $ b_1+b_2<1$. It is straightforward to verify that $q_1\geq 0, ~ (b_1+q_1)q_1=b_2$ and
\begin{equation*}
0\leq \theta=\frac{b_1+\sqrt{b_1^2+4b_2}}{2}=b_1+q_1 <1,
\end{equation*}
where $ \theta<1$ follows from the condition $ b_1+b_2<1. $
Thus it follows from \eqref{lem-tlgeom-1} that
$$
a^{p+1}+q_1a^{p}\leq (b_1+q_1)a^{p}+b_2a^{p-1}+b_3
=\theta(a^{p}+q_1a^{p-1})+b_3,
$$
which implies
\begin{align*}
a^{p+1}\leq    a^{p+1}+q_1a^{p}
\leq   & \theta^p(a^{1}+q_1a^{0})+b_3(1+\theta+\ldots+\theta^{p-1})\\
\leq &   \theta^p\left[a^{1}+(\theta-b_1)a^{0})\right]+\frac{b_3}{1-\theta}.
\end{align*}
Thus the relation \eqref{lem-tlgeom-2} holds.
\end{proof}

\subsection{Guaranteed performance under RIP of order $3k$}

Denote by $ \eta = (\sqrt{5}+1)/2 $ throughout the remaining of this paper. We now prove the guaranteed performance of the proposed algorithms for signal recovery under some assumptions. We first consider the HBHT, to which the main result is stated as follows.

\begin{Thm}\label{theorem-1}
Suppose that the RIC, $\delta_{3k},$ of the measurement matrix $A$ and the parameters $\alpha$ and $ \beta$ obey the bounds
\begin{equation}\label{theorem-1-0}
\delta_{3k}< \frac{\sqrt{5}-1}{2}\approx 0.618, \ \  0\leq\beta<\frac{\eta}{1+\delta_{3k}}-1,\ \
\frac{2(1+\beta)-\eta}{1-\delta_{3k}}<\alpha< \frac{\eta}{1+\delta_{3k}}.
\end{equation} Let $y:=Ax+\nu$  be the measurements of $ x$ with measurement errors $ \nu. $
Then the iterates $\{ x^p\}$ generated by HBHT satisfies
\begin{equation}\label{theorem-1-1}
 {\lVert x_S-x^p\rVert}_2\leq C_1\tau^{p-1} +
C_2\lVert  \nu' \rVert_2,
\end{equation}
where $S=\mathcal{L}_k(x)$, $\nu'=\nu+Ax_{\overline{S}},$ and $C_1,C_2$  are the quantities given as
\begin{equation}\label{theorem-1-1-1}
C_1={\lVert  x_S-x^1 \rVert}_2+({\tau}-{b}){\lVert  x_S-x^0\rVert}_2,\ \
C_2=\frac{\eta\alpha}{1-\tau}\sqrt{1+\delta_{2k}}
\end{equation}
with $\tau:=\frac{b+\sqrt{b^2+4\eta\beta}}{2}.$ The fact $\tau<1$ is ensured under (\ref{theorem-1-0}) and $b$ is given by
\begin{equation}\label{theorem-1-2}
b=\eta(|1-\alpha+\beta|+\alpha\delta_{3k}).
\end{equation}
\end{Thm}

\begin{proof} By \eqref{algorithm-1-1}, we have
 \begin{equation}\label{th-1-pf-1}
u^p-x_S=(1-\alpha+\beta)(x^p-x_S)+\alpha (I-A^ TA)(x^p-x_S)-\beta (x^{p-1}-x_S)
+\alpha A^T \nu',
\end{equation}
where $\nu'=\nu+Ax_{\overline{S}}$. Denote  $V^p:=supp(\mathcal{H}_k(u^p))$.
By using Lemma \ref{lem-optimal} and \eqref{th-1-pf-1}, we obtain
\begin{align}\label{th-1-pf-2}
\lVert x^{p+1}-x_S \rVert_2= & \lVert\mathcal{H}_k(u^p)-x_S \rVert_2  \nonumber \\
\leq &  \eta \lVert  (u^p-x_S )_{S\cup V^p}\rVert_2  \nonumber \\
\leq &  \eta |1-\alpha+\beta|\cdot \lVert x^p-x_S\rVert_2+\eta\alpha \lVert  (I-A^ TA)(x^p-x_S)]_{
S\cup V^p}\rVert_2   \nonumber \\
& \dis +\eta\beta\lVert  x^{p-1}-x_S
\rVert_2+\eta\alpha\lVert( A^T \nu')_{
S\cup V^p} \rVert_2.
\end{align}
Since $| S\cup V^p |\leq 2k$ and $|supp(x^p-x_S)\cup S\cup V^p |\leq 3k$,  by using Lemma \ref{lem-1}, we obtain
\begin{equation}\label{th-1-pf-3}
\lVert  [(I-A^ TA)(x^p-x_S) ]_{
S\cup V^p} \rVert_2\leq\delta_{3k} \lVert x^p-x_S \rVert_2
\end{equation}
and
\begin{equation}\label{th-1-pf-3-1}
 \lVert ( A^T \nu' )_{
S\cup V^p} \rVert_2\leq \sqrt{1+\delta_{2k}} \lVert \nu' \rVert_2.
\end{equation}
 Substituting \eqref{th-1-pf-3}  and \eqref{th-1-pf-3-1} into \eqref{th-1-pf-2} yields
\begin{equation}\label{th-1-pf-4}
\lVert x^{p+1}-x_S \rVert_2
\leq  b  \lVert x^p-x_S \rVert_2+\eta\beta \lVert  x^{p-1}-x_S
 \rVert_2+\eta\alpha\sqrt{1+\delta_{2k}} \lVert  \nu' \rVert_2,
\end{equation}
where $b$ is given by \eqref{theorem-1-2}. The recursive inequality (\ref{th-1-pf-4}) is of the form (\ref{lem-tlgeom-1}) in Lemma \ref{lem-two-level-geometric}. We now point out that the coefficients of the right-hand side of (\ref{th-1-pf-4}) satisfy the condition of Lemma \ref{lem-two-level-geometric}. In fact, suppose that  $\delta_{3k}< =(\sqrt{5}-1)/2=\eta-1$, which implies that $0<\frac{\eta}{1+\delta_{3k}}-1$. Thus the range for $ \beta $ in  \eqref{theorem-1-0}  is well defined, and hence $\frac{2(1+\beta)-\eta}{1-\delta_{3k}}< 1+\beta<\frac{\eta}{1+\delta_{3k}}$. This implies that the range for $ \alpha$ in  \eqref{theorem-1-0}  is also well defined.
Merging \eqref{theorem-1-2} and  \eqref{theorem-1-0} leads to
\begin{align*}\label{th-1-pf-7}
b=&\eta(|1-\alpha+\beta|+\alpha\delta_{3k}) \nonumber \\
=&\left\{
\begin{array}{ll}
\dis\eta[1+\beta-\alpha(1-\delta_{3k})], &\textrm{ if } \frac{2(1+\beta)-\eta}{1-\delta_{3k}}<\alpha\leq 1+\beta, \\
\dis\eta[-1-\beta+\alpha(1+\delta_{3k})],& \textrm{ if }1+\beta<\alpha<\frac{\eta}{1+\delta_{3k}},
\end{array}
\right.  \nonumber \\
<&\left\{
\begin{array}{ll}
\dis\eta\left[1+\beta-\left(\frac{2(1+\beta)-\eta}{1-\delta_{3k}}\right)(1-\delta_{3k})\right], &\textrm{ if } \frac{2(1+\beta)-\eta}{1-\delta_{3k}}<\alpha\leq 1+\beta, \\
\dis\eta\left[-1-\beta+\left(\frac{\eta}{1+\delta_{3k}}\right)(1+\delta_{3k})\right],& \textrm{ if }1+\beta<\alpha<\frac{\eta}{1+\delta_{3k}},
\end{array}
\right. \nonumber \\
 =&\eta(\eta-1-\beta) \nonumber \\
 =& 1-\eta\beta,
\end{align*}
where the last equality follows from the fact $\eta$ is the root of the equation $t^2-t=1.$  The above inequality means $b+\eta \beta <1$ and hence the recursive formula \eqref{th-1-pf-4} satisfies the condition of Lemma  \ref{lem-two-level-geometric}.
Therefore, it follows from Lemma  \ref{lem-two-level-geometric} that
$$
\tau=\frac{b+\sqrt{b^2+4\eta\beta}}{2}<1
$$
 and the bound \eqref{theorem-1-1} holds, where $C_1, C_2$ are given by (\ref{theorem-1-1-1}).
\end{proof}

If the signal $x$ is $k$-sparse and the measurements are accurate, in which case $ x= x_S$ and $ \nu =0$, then the above result implies that
$$
 \| x-x^p\|_2 \leq\tau^{p-1}C_1 \to 0 \textrm{  as } p \to \infty,$$
which implies that the iterates generated by HBHT converges to the sparse signal.

We now establish the main performance result for HBHTP. We first recall a helpful lemma.

\begin{Lem}\emph{\cite{Foucart-2011}} \label{lem-3}
 Given  the measurements $y := Ax+ \nu$ of $x$ and the index set $ S^{p+1},$  the iterate $ x^{p+1}$ generated by the pursuit step (\ref{algorithm-2-3})
obeys
\begin{equation}\label{lem-3-1}
\left\lVert x^{p+1}-x_S\right\rVert_2
\leq \frac{1}{\sqrt{1-(\delta_{2k})^2}}\left\lVert (x^{p+1}-x_S)_{\overline{S^{p+1}}}\right\rVert_2+\frac{ \sqrt{1+\delta_{k}}}{1-\delta_{2k}}\left\lVert \nu' \right\rVert_2,
\end{equation}
where $S=\mathcal{L}_k(x)$ and $\nu'=\nu+Ax_{\overline{S}}. $
\end{Lem}

The main result concerning the guaranteed success of HBHTP is stated as follows.

\begin{Thm}\label{theorem-2}
Suppose that the RIC, $\delta_{3k},$ of the matrix $A$ and the parameters $\alpha$ and $ \beta$ obey
\begin{equation}\label{theorem-2-0}
\delta_{3k}<\frac{1}{\sqrt{3}}\approx 0.577, \ \  0\leq\beta<\frac{\frac{1}{\hat{\eta}}+1}{1+\delta_{3k}}-1,\ \ \frac{1+2\beta-\frac{1}{\hat{\eta}}}{1-\delta_{3k}}<\alpha<\frac{\frac{1}{\hat{\eta}}+1}{1+\delta_{3k}},
\end{equation}
where $\hat{\eta}=\frac{\sqrt{2}}{\sqrt{1-(\delta_{2k})^2}}$. Let $y:=Ax+\nu$ be the measurements of $ x$ with errors $ \nu . $
Then the iterates $\{ x^p\}$ generated by HBHTP  satisfies
\begin{equation}\label{theorem-2-1}
 {\left\lVert x_S-x^p\right\rVert}_2\leq C_3\hat{\tau}^{p-1} +
C_4\left\lVert \nu' \right\rVert_2,
\end{equation}
where $S=\mathcal{L}_k(x)$, $\nu'=\nu+Ax_{\overline{S}}$, and $C_3,C_4$ are given  as
\begin{equation}\label{theorem-2-1-1}
C_3={\left\lVert  x_S-x^1\right\rVert}_2+(\hat{\tau}-\hat{b}){\left\lVert  x_S-x^0\right\rVert}_2,\ \
C_4=\frac{1}{1-\hat{\tau}}\left(\hat{\eta}\alpha \sqrt{1+\delta_{2k}}+\frac{ \sqrt{1+\delta_{k}}}{1-\delta_{2k}}\right)
\end{equation}
with constants $\hat{b}$, $\hat{\tau}$  being given by
\begin{equation}\label{theorem-2-2}
\hat{b}=\hat{\eta}(|1-\alpha+\beta|+\alpha\delta_{3k}),\ \
\hat{\tau}=\frac{\hat{b}+\sqrt{\hat{b}^2+4\hat{\eta}\beta}}{2}
\end{equation}
and $ \hat{\tau}< 1 $ is guaranteed under the condition (\ref{theorem-2-0}).
\end{Thm}

\begin{proof}  Since $S^{p+1}=\mathcal{L}_k(u^p)$ in HBHTP and $S=\mathcal{L}_k(x)$, we have
$$
\|(u^p)_{S^{p+1}}\|_2^2\geq\|(u^p)_S\|_2^2.
$$
Eliminating the entries indexed by $S\cap S^{p+1}$ from the above inequality and taking square root yields
$$
\|(u^p)_{S^{p+1}\setminus S}\|_2\geq\|(u^p)_{S\setminus S^{p+1}}\|_2.
$$
Note that $ (x_S)_{S^{p+1}\setminus S}=0$ and $ (x^{p+1})_{S\setminus S^{p+1}}=0. $ From the inequality above, we have
\begin{align*}
\|(u^p-x_S)_{S^{p+1}\setminus S }\|_2\geq & \|(x_S-x^{p+1}+u^p-x_S)_{S\setminus S^{p+1}}\|_2 \nonumber \\
\geq &    \|(x_S-x^{p+1})_{\overline{ S^{p+1}}}\|_2-\|(u^p-x_S)_{S\setminus S^{p+1}}\|_2,
\end{align*}
where the second inequality follows from the triangular inequality  and  the fact $ (x_S-x^{p+1})_{S\setminus S^{p+1}}= (x_S-x^{p+1})_{\overline{S^{p+1}}} . $
It follows that
\begin{equation}\label{th-2-pf-4}
\begin{array}{rl}
\|(x_S-x^{p+1})_{\overline{ S^{p+1}}}\|_2\leq   &   \|(u^p-x_S)_{S\setminus S^{p+1}}\|_2+\|(u^p-x_S)_{S^{p+1}\setminus S }\|_2\\
\leq   & \sqrt{2\left (\|(u^p-x_S)_{S\setminus S^{p+1}}\|_2^2+\|(u^p-x_S)_{S^{p+1}\setminus S }\|_2^2\right)}\\
=&  \sqrt{2}\|(u^p-x_S)_{S^{p+1}\bigtriangleup S}\|_2,
\end{array}
\end{equation}
where  $S^{p+1}\bigtriangleup S:=(S^{p+1}\setminus S)\cup(S\setminus S^{p+1})$ is the symmetric difference of $S^{p+1}$ and $S.$ The last equality above follows from $(S^{p+1}\setminus S)\cap(S\setminus S^{p+1})=\emptyset$. Note that \eqref{th-1-pf-1}  remains valid for HBHTP.
Merging  \eqref{th-1-pf-1} and \eqref{th-2-pf-4} leads to
\begin{align}\label{th-2-pf-6}
\|(x_S-x^{p+1})_{\overline{ S^{p+1}}}\|_2
\leq & \sqrt{2}\{|1-\alpha+\beta|\cdot\lVert (x^p-x_S)_{S^{p+1}\bigtriangleup S } \rVert_2+\alpha \lVert ( A^T \nu' )_{S^{p+1}\bigtriangleup S } \rVert_2   \nonumber \\
& \dis +\alpha\lVert  [(I-A^ TA)(x^p-x_S)]_{S^{p+1}\bigtriangleup S }\rVert_2+\beta \lVert  (x^{p-1}-x_S)_{S^{p+1}\bigtriangleup S }\rVert_2\}.
\end{align}
  Since $|S^{p+1}\bigtriangleup S|\leq 2k$ and $|(S^{p+1}\bigtriangleup S)\cup supp(x^p-x_S)|\leq 3k$, by using  Lemma \ref{lem-1}, one has
 \begin{equation}\label{th-2-pf-6-1}
 \lVert [(I-A^ TA)(x^p-x_S)]_{
 S^{p+1}\bigtriangleup  S}\rVert_2\leq\delta_{3k}\left\lVert x^p-x_S\right\rVert_2
 \end{equation}
 and
 \begin{equation}\label{th-2-pf-6-2}
 \lVert( A^T \nu')_{
 S^{p+1}\bigtriangleup S}\rVert_2\leq \sqrt{1+\delta_{2k}} \lVert \nu' \rVert_2.
 \end{equation}
 Combining (\ref {th-2-pf-6})-(\ref{th-2-pf-6-2}) leads to
\begin{align*}
\|(x_S-x^{p+1})_{\overline{ S^{p+1}}}\|_2
\leq &  \sqrt{2}\{(|1-\alpha+\beta|+\alpha\delta_{3k})\lVert x^p-x_S\rVert_2+\alpha\sqrt{1+\delta_{2k}} \lVert\nu' \rVert_2 \nonumber \\
& \dis +\beta \lVert  x^{p-1}-x_S\rVert_2\}.
\end{align*}
Merging the inequality above and  \eqref{lem-3-1} in Lemma \ref{lem-3}, we obtain
\begin{equation}\label{th-2-pf-8}
\left\lVert x^{p+1}-x_S \right\rVert_2
\leq\hat{b}\left\lVert x^p-x_S\right\rVert_2+\hat{\eta}\beta\left\lVert  x^{p-1}-x_S
\right\rVert_2+\left(1-\hat{\tau}\right)C_4\left\lVert \nu' \right\rVert_2,
\end{equation}
where $\hat{\eta},\hat{b},\hat{\tau},C_4$ are given exactly as in Theorem \ref{theorem-2}.

Since $\delta_{2k}\leq\delta_{3k}<\frac{1}{\sqrt{3}}$, we have
 $\sqrt{2}\delta_{3k}<\sqrt{1-(\delta_{3k})^2}\leq\sqrt{1-(\delta_{2k})^2}=\sqrt{2}/\hat{\eta},$ which implies that $0<\frac{\frac{1}{\hat{\eta}}+1}{1+\delta_{3k}}-1.$ Therefore,  the range for $ \beta $  in \eqref{theorem-2-0} is well defined, which also implies that $$\frac{1+2\beta-\frac{1}{\hat{\eta}}}{1-\delta_{3k}}<1+\beta
 <\frac{\frac{1}{\hat{\eta}}+1}{1+\delta_{3k}},$$
and hence the range for $ \alpha$ in \eqref{theorem-2-0}  is also well defined.
Thus it follows from  \eqref{theorem-2-0} and \eqref{theorem-2-2} that
\begin{align*}
\hat{b}= &\hat{\eta}(|1-\alpha+\beta|+\alpha\delta_{3k})\\
= &\left\{
\begin{array}{ll}
\dis\hat{\eta}[1+\beta-\alpha(1-\delta_{3k})],\ \ &\frac{1+2\beta-\frac{1}{\hat{\eta}}}{1-\delta_{3k}}<\alpha\leq 1+\beta, \\
\dis\hat{\eta}[-1-\beta+\alpha(1+\delta_{3k})],\ \ &1+\beta< \alpha <\frac{\frac{1}{\hat{\eta}}+1}{1+\delta_{3k}},
\end{array}
\right.\\
<& 1-\hat{\eta}\beta,
\end{align*}
i.e., $ b+\hat{\eta} \beta <1. $  Therefore, applying Lemma \ref{lem-two-level-geometric} to the recursive relation \eqref{th-2-pf-8}, we immediately conclude that $\hat{\tau} <1 $  and the desired estimation \eqref{theorem-2-1} holds.
\end{proof}

When the measurements are accurate and the signal is $k$-sparse, Theorem \ref{theorem-2} implies that the sequence $ \{x^p\}$  produced by the HBHTP must converge to the signal as $ p \to \infty.$ That is,  the algorithm exactly recovers the signal in this case. It is also worth pointing out that computing the RIC of a matrix is generally difficult. Thus in practical applications, we do not require that the parameters $\alpha $ and $\beta $ be chosen to strictly meet the condition  \eqref{theorem-1-0} or  \eqref{theorem-2-0}. These parameters can be set to roughly satisfy these conditions, for instance,
\begin{equation}\label{rem-parameter-1}
0\leq\beta<\eta-1,\ \
2+2\beta-\eta<\alpha<\eta,
\end{equation}
where $\eta=(\sqrt{5}+1)/2$.
As examples, we may simply set $\alpha\in[0.4+2\beta,1.6]$ and $  \beta\in(0,0.6]$ in  HBHT for simplicity, and set $\alpha\in[0.3+2\beta,1.7] $ and $ \beta\in(0,0.7]$ in HBHTP.

\subsection{Guaranteed performance under RIP of order $2k$}

Motivated by the idea of decomposition method in \cite{Zhao-2020-Analysis}, we first establish a helpful inequality, based on which the guaranteed performance of the proposed algorithms can be characterized immediately in terms of RIP of order $ 2k. $

\begin{Lem}\label{lem-2}
Let $x$ and $z$ be two $k$-sparse vectors, $S=supp(x)$ and $S^*\subseteq N$ be an index set. If $|S\cup S^*|\leq 2k$, then
\begin{equation}\label{lem-2-1}
\left\lVert \left [(I-A^ T A)(x-z)\right]_{S\cup S^*}\right\rVert_2\leq \sqrt{3}  \delta_{2k} {\left\lVert x-z \right\rVert}_2.
\end{equation}
\end{Lem}

\begin{proof} In this proof, we denote by ${\bf e}=(1,1,\ldots,1)^T$ the $n$-dimensional vector of ones and we use the symbol $u\otimes v:=(u_1v_1,\ldots, u_nv_n)^T$ to denote the Hadamard product of two vectors $u, v\in \mathbb{R}^n. $  Let $ x, z, S, S^*$ be specified as in this lemma.
Let $\hat{\omega}\in\{0,1\}^n$ be a $2k$-sparse binary vector such that  $S\cup S^*\subseteq supp(\hat{\omega})$.
 We  partition $\hat{\omega}$  into two $k$-sparse binary vectors ${\omega}'$ and ${\omega}''$, i.e., $\hat{\omega}={\omega}'+{\omega}''$, where $supp({\omega}')\cap supp({\omega}'')=\emptyset$. The following relation holds for any  $u \in \mathbb{R}^n: $
\begin{equation}\label{lem-2-pf-1}
\left\lVert u\otimes\hat{\omega}\right\rVert_2^2=\left\lVert u\otimes{\omega}'\right\rVert_2^2+\left\lVert u\otimes{\omega}''\right\rVert_2^2.
\end{equation}
  Note that $x-z$ can be decomposed into two sparse vectors  $v^{(1)}$ and  $v^{(2)}$, i.e., $x-z=v^{(1)}+v^{(2)}$,
where $v^{(1)} = (x -z) \otimes\hat{\omega}$  is a $2k$-sparse vector and  $v^{(2)} = (x - z) \otimes(\bf e-\hat{\omega})$ is  a $k$-sparse
vector since $S\subseteq supp(\hat{\omega})$ and $ z$ is $k$-sparse. It is easy to see that
\begin{align}\label{lem-2-pf-2}
\left\lVert \left [(I-A^ T A)(x-z)\right]_{S\cup S^*}\right\rVert_2
\leq &\lVert  [(I-A^ T A)(x-z)]_{supp(\hat{\omega})}\rVert_2 \nonumber \\
=& \lVert  [(I-A^ T A)(v^{(1)}+v^{(2)}) ]\otimes\hat{\omega} \rVert_2 \nonumber \\
\leq &
\lVert  [(I-A^ T A)v^{(1)} ]\otimes\hat{\omega} \rVert_2
+ \lVert  [(I-A^ T A)v^{(2)} ]\otimes\hat{\omega} \rVert_2.
\end{align}
Since  $supp(v^{(1)})\subseteq supp(\hat{\omega})$, we have $|supp(v^{(1)})\cup supp(\hat{\omega})|\leq 2k$. It follows from Lemma \ref{lem-1} (i) that
\begin{equation}\label{lem-2-pf-3}
\lVert [(I-A^ T A)v^{(1)}]\otimes\hat{\omega}\rVert_2=\lVert  [(I-A^ T A)v^{(1)}]_{supp(\hat{\omega})}\rVert_2\leq  \delta_{2k} {\lVert v^{(1)} \rVert}_2.
\end{equation}
Since $|supp(v^{(2)})\cup supp({\omega}')|\leq 2k$ and $|supp(v^{(2)})\cup supp({\omega}'')|\leq 2k$, by using (\ref{lem-2-pf-1}) and Lemma \ref{lem-1} (i), we obtain
\begin{align*}\label{lem-2-pf-4}
\lVert  [(I-A^ T A)v^{(2)}]\otimes\hat{\omega}\rVert_2^2= & \lVert [(I-A^ T A)v^{(2)}]\otimes{\omega}'\rVert_2^2+\lVert  [(I-A^ T A)v^{(2)}]\otimes{\omega}''\rVert_2^2   \nonumber \\
\leq &  2(\delta_{2k} )^2{\lVert v^{(2)} \rVert}_2^2,
\end{align*}
i.e.,
\begin{equation}\label{lem-2-pf-5}
\lVert  [(I-A^ T A)v^{(2)} ]\otimes\hat{\omega} \rVert_2\leq  \sqrt{2}\delta_{2k} {\lVert v^{(2)} \rVert}_2.
\end{equation}
Combining \eqref{lem-2-pf-2}, \eqref{lem-2-pf-3} and  \eqref{lem-2-pf-5} yields
\begin{align}\label{lem-2-pf-6}
\left\lVert \left [(I-A^ T A)(x-z)\right]_{S\cup S^*}\right\rVert_2
\leq &   \delta_{2k}\left({\lVert v^{(1)} \rVert}_2+\sqrt{2} {\lVert v^{(2)}\rVert}_2\right)    \nonumber\\
\leq  & \sqrt{3} \delta_{2k}\sqrt{{\left\lVert v^{(1)} \right\rVert}_2^2+{\left\lVert v^{(2)} \right\rVert}_2^2}
= \sqrt{3} \delta_{2k}{\left\lVert x-z \right\rVert}_2,
\end{align}
where the second inequality follows from the fact $a+\sqrt{2}c\leq\sqrt{3(a^2+c^2)}$ for any $a,c\geq 0$.
\end{proof}

  According to Lemma \ref{lem-2}, the term $ \delta_{3k}$ in bounds \eqref{th-1-pf-3} and \eqref{th-2-pf-6-1} can be  replaced with $\sqrt{3}  \delta_{2k}. $ Thus we immediately obtain the theoretical performance results for the proposed algorithms in terms of RIP of order $2k. $ 	

\begin{Cor}\label{remark-1} Let $ y:= Ax + \nu$ be the inaccurate measurements of $ x.$ If the RIC, $\delta_{2k}, $  of the matrix $ A$ and the parameters $ \alpha$ and  $ \beta$ in HBHT satisfy the conditions:
\begin{equation}\label{rem-1-0}
\delta_{2k}<\frac{\sqrt{5}-1}{2\sqrt{3}} \approx 0.356,\ \ 0\leq\beta<\frac{\eta}{1+\sqrt{3}  \delta_{2k}}-1,\ \  \frac{2(1+\beta)-\eta}{1-\sqrt{3}  \delta_{2k}}<\alpha< \frac{\eta}{1+\sqrt{3}  \delta_{2k}},
\end{equation}
Then the conclusion of Theorem \ref{theorem-1} remains valid, with constants $\tau, C_1, C_2$ being defined the same way therein except
$b=\eta(|1-\alpha+\beta|+\alpha\sqrt{3}  \delta_{2k}). $
\end{Cor}

\begin{Cor}\label{remark-1} Let $ y:= Ax + \nu$ be the inaccurate measurements of $ x.$ If the RIC, $\delta_{2k}, $ of the matrix $ A$ and the parameters $ \alpha$ and  $ \beta$ in HBHTP satisfy the conditions:
 \begin{equation}\label{rem-2-1}
\delta_{2k}<\frac{1}{\sqrt{7} }\approx 0.377,\ \ 0\leq\beta<\frac{\frac{1}{\hat{\eta}}+1}{1+\sqrt{3}  \delta_{2k}}-1,\ \   \frac{1+2\beta-\frac{1}{\hat{\eta}}}{1-\sqrt{3}\delta_{2k}}<\alpha<\frac{\frac{1}{\hat{\eta}}+1}{1+\sqrt{3}\delta_{2k}},
\end{equation}
Then the conclusion of Theorem \ref{theorem-2} remains valid, with constants $\hat{\tau}, C_3, C_4$ being defined the same way therein except
$\hat{b}=\hat{\eta}(|1-\alpha+\beta|+\alpha\sqrt{3}  \delta_{2k}).
$
\end{Cor}

 \begin{rem}\label{remark-3}
 \emph{According to Proposition 6.6 in \cite {Foucart-2013}, one has the relation $\delta_{3k}\leq 3\delta_{2k}.$ If we use this relation to derive an upper bound for the left-hand side of \eqref{lem-2-1}, then the resulting bound would be too loose. The bound \eqref{lem-2-1} established here is much tighter, and thus it leads to a desired strong  result.
We summarize the best known conditions for guaranteed performance of several compressed sensing algorithms in Table  \ref{table-RIP}.
\begin{table}[htbp]
 \centering
 \caption{RIP-based bounds}\label{table-RIP}
 \vspace{0.2cm}
 \begin{tabular}{|c|c|c|c|c|c|}
\hline
Algorithms & {IHT}\cite{Zhao-2020-Improved} & IHT$^\mu$/GDS\cite{Foucart-2013,garg2009gradient}& HBHT & HTP\cite{Foucart-2011} & HBHTP \\
\hline
$\delta_{3k}<\delta_*$&{$0.618$}& &$0.618$&$0.577$&$0.577$\\
\hline
$\delta_{2k}<\delta_*$& &$0.333$&$0.356$& &$0.377$\\
\hline
 \end{tabular}
 \end{table}
Our analysis indicates that the RIP-based bounds for the performance guarantee of HBHT and HBHTP  can be the same as the best known  bounds for IHT and HTP, respectively.  Similar to the sufficient condition $\delta_{2k} < 0.333 $ for the performance guarantee of $IHT^\mu$ in \cite{Foucart-2013,garg2009gradient}, the more relaxed condition $\delta_{2k} < 0.356 $ is obtained for the algorithm HBHT in this paper.  It is also interesting to observe that the sufficient condition $\delta_{2k}< 0.377 $ for HBHTP is less restrictive than that of HBHT, while the conditions in terms of $\delta_{3k}$ for the two algorithms go other way round.}
\end{rem}

\subsection{Finite convergence}

From accurate measurements, the HBHTP can exactly recover a $k$-sparse signal in a finite number of iterations. The iteration complexity is given in the next result.

\begin{Thm}\label{corollary-1}
Suppose that the RIC, $\delta_{3k},$ of the measurement matrix $A$ and the algorithmic parameters $\alpha$ and $\beta$ in HBHTP satisfy  the condition \eqref{theorem-2-0}.
Then any $k$-sparse signal $x$ with $ \|x\|_0=k$ can be exactly  recovered by HBHTP from accurate measurements $y := Ax$ in at most
 \begin{equation}\label{corollary-1-1}
p^*=\left\lceil \frac{\log \left(\frac{\sqrt{2}C_3}{\hat{\eta}\mu}\right)}{\log \left(1/\hat{\tau}\right)}\right\rceil+1
\end{equation}
 iterations, where $
 \mu=\min_{x_i\not=0}|x_i|,
$ and $\hat{\eta},\hat{\tau},C_3$ are given in Theorem \ref{theorem-2}.

\end{Thm}

\begin{proof} Denote by $ S={\cal L}_k(x).$
 For any $t\in \bar{S}$, by using the definition of $ u^p$ in HBHTP, which is defined as \eqref{algorithm-1-1}, and noting that  $ x_t=0,$  we have
  \begin{align*}
 |(u^p)_t|=&   \left|x_t+(1-\alpha+\beta)(x^p-x)_t+\alpha \left[(I-A^ TA)(x^p-x)\right]_t-\beta (x^{p-1}-x)_t\right|\\
 \leq   & |1-\alpha+\beta|\cdot|(x^p-x)_t|+\alpha \left|\left[(I-A^ TA)(x^p-x)\right]_t\right|+\beta|(x^{p-1}-x)_t|,
 \end{align*}
and for any $ s\in S, $ we have
\begin{align*}
 \begin{array}{rl}
|(u^p)_s|=& \left|x_s+(1-\alpha+\beta)(x^p-x)_s+\alpha \left[(I-A^ TA)(x^p-x)\right]_s-\beta (x^{p-1}-x)_s\right|\\
\geq& \mu-|1-\alpha+\beta|\cdot|(x^p-x)_s|-\alpha \left|\left[(I-A^ TA)(x^p-x)\right]_s\right|-\beta|(x^{p-1}-x)_s|, \\
\end{array}
\end{align*}
where  $\mu=\min_{x_i\not=0}|x_i|.$  Combining the above two inequalities leads to
 \begin{align}\label{cor-1-pf-3}
|(u^p)_t|    - &|(u^p)_s|+\mu  \nonumber \\
  \leq  &   |1-\alpha+\beta|\cdot [|(x^p-x)_t|+|(x^p-x)_s| ]+\beta [|(x^{p-1}-x)_t|+|(x^{p-1}-x)_s|]   \nonumber \\
&  \dis+\alpha \{|[(I-A^ TA)(x^p-x)] _t|+|[(I-A^ TA)(x^p-x)] _s|\}  \nonumber \\
 \leq  & \sqrt{2}(|1-\alpha+\beta|\cdot\|(x^p-x)_{\{s,t\}}\|_2+\alpha \|[(I-A^ TA)(x^p-x)]_{\{s,t\}}\|_2 \nonumber \\
 & \dis +\beta\|(x^{p-1}-x)_{\{s,t\}}\|_2).
\end{align}
Since  $s\in S=supp(x)$, we have  $|supp(x^p-x)\cup \{s,t\} |\leq 2k+1\leq 3k$. Using \eqref{theorem-2-2} and Lemma \ref{lem-1} (i), we  obtain
 \begin{align}\label{cor-1-pf-4}
|(u^p)_t|-|(u^p)_s|+\mu
\leq & \sqrt{2}\left((|1-\alpha+\beta|+\alpha \delta_{3k})\|x^p-x\|_2+\beta\|x^{p-1}-x\|_2\right) \nonumber \\
=&\frac{\sqrt{2}}{\hat{\eta}}\left(\hat{b}\|x^p-x\|_2+\hat{\tau}
(\hat{\tau}-\hat{b})\|x^{p-1}-x\|_2\right) \nonumber \\
\leq & \frac{\sqrt{2}}{\hat{\eta}}\hat{\tau}\left(\|x^p-x\|_2
+(\hat{\tau}-\hat{b})\|x^{p-1}-x\|_2\right),
\end{align}
where $\hat{\eta},\hat{\tau}, \hat{b}$ are given  in Theorem \ref{theorem-2}, and the last inequality above follows from the fact $\hat{b} < \hat{\tau}. $
Since $x$ is a $k$-sparse vector  and $\nu=0$, then  $\nu'=\nu+Ax_{\overline{S}}=0$. Hence, \eqref{th-2-pf-8} becomes
$$
\left\lVert x^{p+1}-x\right\rVert_2
\leq\hat{b}\left\lVert x^p-x\right\rVert_2+\hat{\eta}\beta\left\lVert  x^{p-1}-x
\right\rVert_2.
$$
With the aid  of \eqref{theorem-2-2} and note that $ \hat{\tau}(\hat{\tau}-\hat{b})= \hat{\eta} \hat{\beta},$  the inequality above can be further rewritten as
$$
\lVert x^{p+1}-x \rVert_2+(\hat{\tau}-\hat{b})\lVert x^{p}-x \rVert_2
\leq\hat{\tau}(\|x^p-x\|_2+(\hat{\tau}-\hat{b})\|x^{p-1}-x\|_2).
$$
Thus \eqref{cor-1-pf-4} reduces to
 $$
|(u^p)_t|-|(u^p)_s|+\mu
\leq\dis\frac{\sqrt{2}}{\hat{\eta}}C_3(\hat{\tau})^{p},
$$
where $C_3$ is given by \eqref{theorem-2-1-1}. After $p^*$ iterations, where $p^*$ is given by \eqref{corollary-1-1}, one must have that
 $$
|(u^{p^*})_t|-|(u^{p^*})_s|+\mu
\leq\dis\frac{\sqrt{2}}{\hat{\eta}}C_3(\hat{\tau})^{{p^*}}< \mu,
$$
where the second inequality holds due to the definition of $ p^*$ in  \eqref{corollary-1-1}.
It  implies that $ |(u^{p^*})_t|<|(u^{p^*})_s|$ for any  $s\in S$ and $ t\in \overline{S}$. This means $ S= {\cal L}_k (u^{p^*}).$  Note that $S^{p^*+1}=\mathcal{L}_k(u^{p^*})$ at the $p^*$-th iteration of HBHTP. Thus  at the $ p^*$-th iteration, one has $ S= S^{p^*+1}.$ Under the RIP condition which implies that any $k$ columns of $A$ are linearly independent, the system $ y= Az$ has at most one $k$-sparse solution. Therefore,  $x^{p^*+1}=x, $ i.e.,  the HBHTP successfully recover the $k$-sparse signal $x$ after finite number of iterations.
\end{proof}

\section{Stability of HBHT and HBHTP}\label{stability}
An efficient compressed sensing algorithm should be able to recover  signals in a stable manner in the sense that when the problem data (e.g., signal, measurement,  noise level) admits a slight change, the quality of signal recovery can  still  be guaranteed and the recovery error is still under control.
In this section, we establish a stability result for HBHT and HBHTP, respectively.
Recall that for given two integer numbers $s$ and $q$, the symbol $ \sigma_s(x)_q $ denotes the error (in terms of $ \ell _q$-norm) of the best $s$-term approximation of the vector $ x, $ i.e.,
$
\sigma_s(x)_q:=\inf\{\|x-z\|_q:  \|z\|_0\leq s\}.
 $
 We first give the following inequalities taken from \cite{Foucart-2013} (see, Theorem 2.5 and Lemma 6.10 therein).

 \begin{Lem} \label{4242} (i) For any $ z\in \mathbb{R}^n, $  $ \sigma_s(z)_2\leq \frac{1}{2\sqrt{s}}\| z\|_1. $
  (ii) For any $u,v \in \mathbb{R}^n$ satisfying $
   \underset{1\leq i\leq n}{\max}|u_i|\leq \underset{1\leq i\leq n}{\min}|v_i|,$
  one has $\| u\|_2 \leq \frac{1}{\sqrt{n}} \|v\|_1. $
 \end{Lem}

We now  establish a lemma which is an modification of Lemma 6.23 in \cite{Foucart-2013}  with $\ell_2$-norm.

\begin{Lem}\label{lem-4}
Given $x,x'\in\mathbb{R}^n$, $A\in\mathbb{R}^{m\times n}$, $\nu\in\mathbb{R}^m$  and the scalars  $\phi>0$ and $ \xi\geq 0.$  Let  $x'$ be a  $k$-sparse vector ($k\geq 2$) and $T:=\mathcal{L}_k(x)$. If
\begin{equation}\label{lem-4-2}
\|x_T-x'\|_2\leq \phi\|Ax_{\overline{T}}+\nu\|_2+\xi,
\end{equation}
then
\begin{equation}\label{lem-4-3}
\|x-x'\|_2\leq \frac{1+2\phi\sqrt{1+\delta_j}}{2\sqrt{j}}\sigma_j(x)_1+\phi\|\nu\|_2+\xi,
\end{equation}
where $j=\lfloor \frac{k}{2} \rfloor .$
\end{Lem}

\begin{proof} There are only two cases.

Case I: $k\geq 2 $ is an odd integer number.  In this case, $|T|=k=2j+1$.
Denote $S_0:={\cal L}_{j+1}(x)\subset T$ and $S_1:=T\setminus S_0$.  It is not difficult to see that
\begin{equation}\label{lem-4-pf-1}
\|x_{\overline{T}}\|_2=\sigma_{j}(x_{\overline{S_0}})_2\leq \frac{1}{2\sqrt{j}}\|x_{\overline{S_0}}\|_1= \frac{1}{2\sqrt{j}}\sigma_{j+1}(x)_1,
\end{equation}
where  the first and final equalities follow from the definition of $ S_0, T$ and $ \sigma_s(\cdot)_q$, and the inequality in between follows from Lemma \ref{4242}(i).
There exists an integer $r\geq 2$ such that  $\overline{T}$ can be partitioned as  $\overline{T}=\bigcup\limits_{l=2}^r S_l$, where
$$S_2=L_j(x_{\overline{T}}),\ S_3=L_j(x_{\overline{T\cup S_2}}),\ldots,S_{r-1}=L_j(x_{\overline{T\cup S_2\cup\ldots \cup  S_{r-2}}}),\ S_r=\overline{T\cup S_2\cup\ldots \cup S_{r-1}}$$ with cardinalities $|S_l|=j$ for $j=2,\dots , r-1$ and $|S_r|\leq j.$ Using the triangular inequality together with  \eqref{def-RIC-1}, we obtain
\begin{equation}\label{lem-4-pf-2}
\|Ax_{\overline{T}}+\nu\|_2\leq \sum_{l= 2}^r\|Ax_{S_l}\|_2+\|\nu\|_2\leq \sqrt{1+\delta_j}\sum_{l= 2}^r\|x_{S_l}\|_2+\|\nu\|_2.
\end{equation}
Based on Lemma \ref{4242} (ii),  we observe that
$$
\|x_{S_l}\|_2\leq \frac{1}{\sqrt{j}} \|x_{S_{l-1}}\|_1,\ \ 2\leq l \leq r.
$$
This together with \eqref{lem-4-pf-2} implies that
\begin{equation}\label{lem-4-pf-4}
\|Ax_{\overline{T}}+\nu\|_2\leq \sqrt{\frac{1+\delta_j}{j} }\sum_{l= 1}^{r-1}\|x_{S_l}\|_1+\|\nu\|_2\leq\sqrt{\frac{1+\delta_j}{j} }\|x_{\overline{S_0}}\|_1+\|\nu\|_2=\sqrt{\frac{1+\delta_j}{j} }\sigma_{j+1}(x)_1+\|\nu\|_2.
\end{equation}
Combining \eqref{lem-4-pf-1}, \eqref{lem-4-2} with \eqref{lem-4-pf-4} leads to
\begin{equation}\label{lem-4-pf-5}
\|x-x'\|_2\leq\|x_{\overline{T}}\|_2+ \|x_{T}-x'\|_2\leq
\frac{1}{2\sqrt{j}}\sigma_{j+1}(x)_1+\phi\sqrt{\frac{1+\delta_j}{j} }\sigma_{j+1}(x)_1+\phi\|\nu\|_2+\xi.
\end{equation}

Case II: $k\geq 2$ is an even integer number. In this case, $|T|=k=2j$.
Denote $S_0:={\cal L}_{j}(x)\subset T$.  Repeating the argument in Case I and using the relation $\|x_{\overline{S_0}}\|_1= \sigma_{j}(x)_1$ for this case,  we obtain the following relation:
\begin{equation}\label{lem-4-pf-6}
\|x-x'\|_2\leq
\frac{1}{2\sqrt{j}}\sigma_{j}(x)_1+\phi\sqrt{\frac{1+\delta_j}{j} }\sigma_{j}(x)_1+\phi\|\nu\|_2+\xi.
\end{equation}
Note that $\sigma_{j+1}(x)_1\leq\sigma_{j}(x)_1$. Both \eqref{lem-4-pf-5}  and  \eqref{lem-4-pf-6} imply the desired relation \eqref{lem-4-3}  for any positive integer $k\geq 2$.
\end{proof}

By using  Lemma \ref{lem-4}, the main result on the stability of HBHT  can be stated as follows.

\begin{Thm}\label{theorem-3}
Suppose that the RIC, $\delta_{3k} $ ($k \geq 2),$ of the matrix $A$ and the parameters $\alpha$ and $\beta$ satisfy the conditions in \eqref{theorem-1-0}. Let $y:=Ax+\nu$ be the measurements of  $x $ with measurement errors $ \nu.$
Then the sequence  $\{ x^p\},$ generated by HBHT with initial points $ x^1=x^0= 0,$ satisfies
\begin{equation}\label{theorem-3-1}
 {\left\lVert x-x^p\right\rVert}_2\leq \frac{1+2C_2\sqrt{1+\delta_j}}{2\sqrt{j}}\sigma_j(x)_1+C_2\|\nu\|_2+({\tau}-{b}+1)\tau^{p-1}\|x\|_2,
\end{equation}
where $j=\lfloor \frac{k}{2}\rfloor $ and $C_2,\tau, b$ are given in Theorem \ref{theorem-1}.
\end{Thm}
\begin{proof}According to \eqref{theorem-1-1}, we know
$$
{ \lVert x_S-x^p \rVert}_2\leq C_1\tau^{p-1} +
C_2 \lVert  \nu' \rVert_2=C_1\tau^{p-1} +
C_2 \lVert  \nu+Ax_{\overline{S}} \rVert_2,
$$
where $S=\mathcal{L}_k(x)$. This is the form of \eqref{lem-4-2}  in Lemma \ref{lem-4} with  $x'=x^p, \phi=C_2, \xi=C_1\tau^{p-1}$ and $T=S$. Hence, it follows from \eqref{lem-4-3} that
\begin{equation}\label{th-3-pf-1}
\|x-x^p\|_2\leq \frac{1+2C_2\sqrt{1+\delta_j}}{2\sqrt{j}}\sigma_j(x)_1+C_2\|\nu\|_2+C_1\tau^{p-1}.
\end{equation}
Substituting $x^1=x^0=0$ into \eqref{theorem-1-1-1} yields
\begin{equation}\label{th-3-pf-2}
C_1={\lVert  x_S \rVert}_2+({\tau}-{b}){\lVert  x_S \rVert}_2\leq ({\tau}-{b}+1){\lVert  x \rVert}_2.
\end{equation}
Combining \eqref{th-3-pf-1} and \eqref{th-3-pf-2} leads to the desired estimation \eqref{theorem-3-1}.
\end{proof}

By a similar  proof to the above, we obtain the stability result for HBHTP.

\begin{Thm}\label{theorem-4}
Suppose that  the RIC, $\delta_{3k}(k \geq 2),$ of the matrix $A$ and the parameters $\alpha$ and $\beta$ satisfy the conditions in \eqref{theorem-2-0}. Let $y:=Ax+\nu$ be the measurements of $ x.$
Then the iterates $\{ x^p\},$ generated by HBHTP with initial points $x^1=x^0=0,$ satisfies
\begin{equation}\label{theorem-4-1}
 {\lVert x-x^p\rVert}_2\leq \frac{1+2C_4\sqrt{1+\delta_j}}{2\sqrt{j}}\sigma_j(x)_1+C_4\|\nu\|_2+(\hat{\tau}-\hat{b}+1){\hat{\tau}}^{p-1}\|x\|_2,
\end{equation}
where $j=\lfloor \frac{k}{2}\rfloor $ and $C_4,\hat{\tau},\hat{b}$ are given in Theorem \ref{theorem-2}.
\end{Thm}

From \eqref{theorem-3-1} and \eqref{theorem-4-1}, we see that the recovery error $ \|x-x^p\|_2$  can be controlled and can be measured in terms of $\sigma_j(x)_1,$ measurement errors, and the number of iterations performed.  These results claims that a slight variance of these factors will not significantly affect the recovery error, and that if the signal is $j=\lfloor k/2\rfloor$-compressible (i.e.,  $\sigma_j(x)_1$ is small) and if the measurements are accurate enough, then the signal will be recovered by the proposed algorithms provided that enough number of iterations are performed.

\section{Numerical experiments}\label{simulation}

All mentioned experiments in this section were performed on a PC with the processor Intel(R) Core(TM) i7-10700 CPU @ 2.90GHz and 16GB memory. In these experiments, the measurement matrices $A\in\mathbb{R}^{m\times n}$ are Gaussian random matrices whose entries are independent and identically distributed (iid) and follow the standard normal distribution ${\mathcal N}(0,1)$ in Section \ref{recovery} and  $\mathcal N(0,m^{-1})$  in Section \ref{NE-phase-transition}, respectively. All sparse vectors  $x^*\in\mathbb{R}^{ n}$ are also randomly generated, whose nonzero entries  are iid and follow ${\mathcal N}(0,1)$ and the position of nonzero entries  follows the uniform distribution.

\subsection{Comparison of performance}\label{recovery}

We first demonstrate some numerical results on recovery success rates of HBHT and HBHTP and the average number of iterations and CPU time required by these algorithms  to achieve the recovery success of sparse signals. We compare their performances with the iterative algorithms OMP, SP, CoSaMP, HTP and IHT. We let HBHT and HBHTP start from $x^0=x^1= 0$ and other iterative algorithms start from  $x^0= 0.$
 The size of the matrices  in this experiment is $400\times 800 .$ All iterative algorithms  are allowed to perform up to 50 iterations (which is set as the maximum number of iterations in our experiments), except for OMP which, by its structure, is performed exactly $k$ iterations, equal to the sparsity level of the target signal $x^*.$  For every given sparsity level $k$, 100 random examples of $(A,x^*)$ are generated to estimate the success rates of algorithms. An individual recovery is called success if the solution produced by an algorithm satisfies the criterion
\begin{equation}\label{recov-criter}
\|x^p-x^*\|_2/\|x^*\|_2\leq 10^{-3}.
\end{equation}
Let us first compare the algorithms in the case of $ A$ being un-normalized.

 \begin{figure}[htbp]
  \centering
  \subfigure[Accurate measurements]{
  \begin{minipage}[t]{0.45\linewidth}
  \centering
  \includegraphics[width=\textwidth,height=0.8\textwidth]{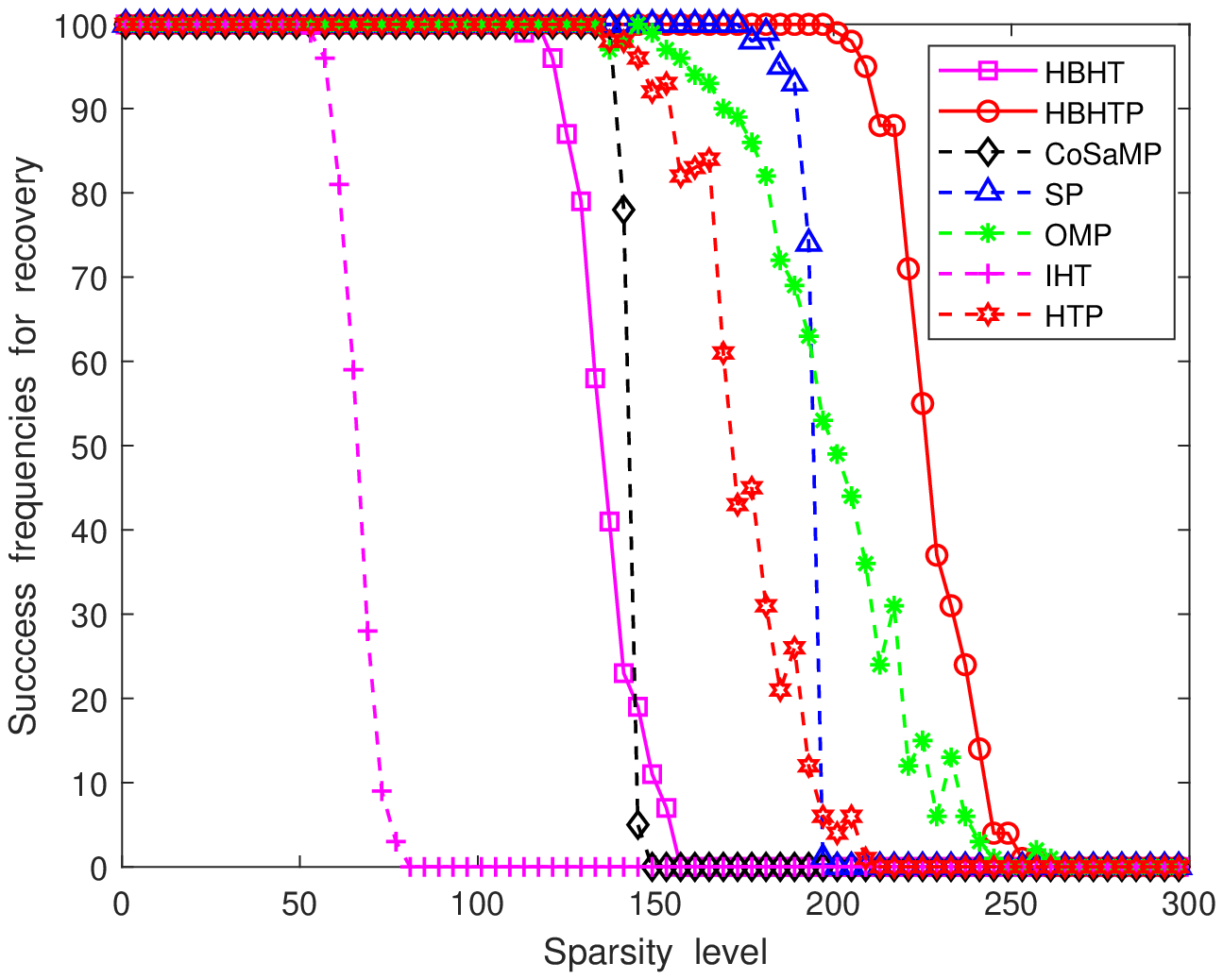}
  \end{minipage}
  }
  \subfigure[ Inaccurate measurements.]{
  \begin{minipage}[t]{0.45\linewidth}
  \centering
  \includegraphics[width=\textwidth,height=0.8\textwidth]{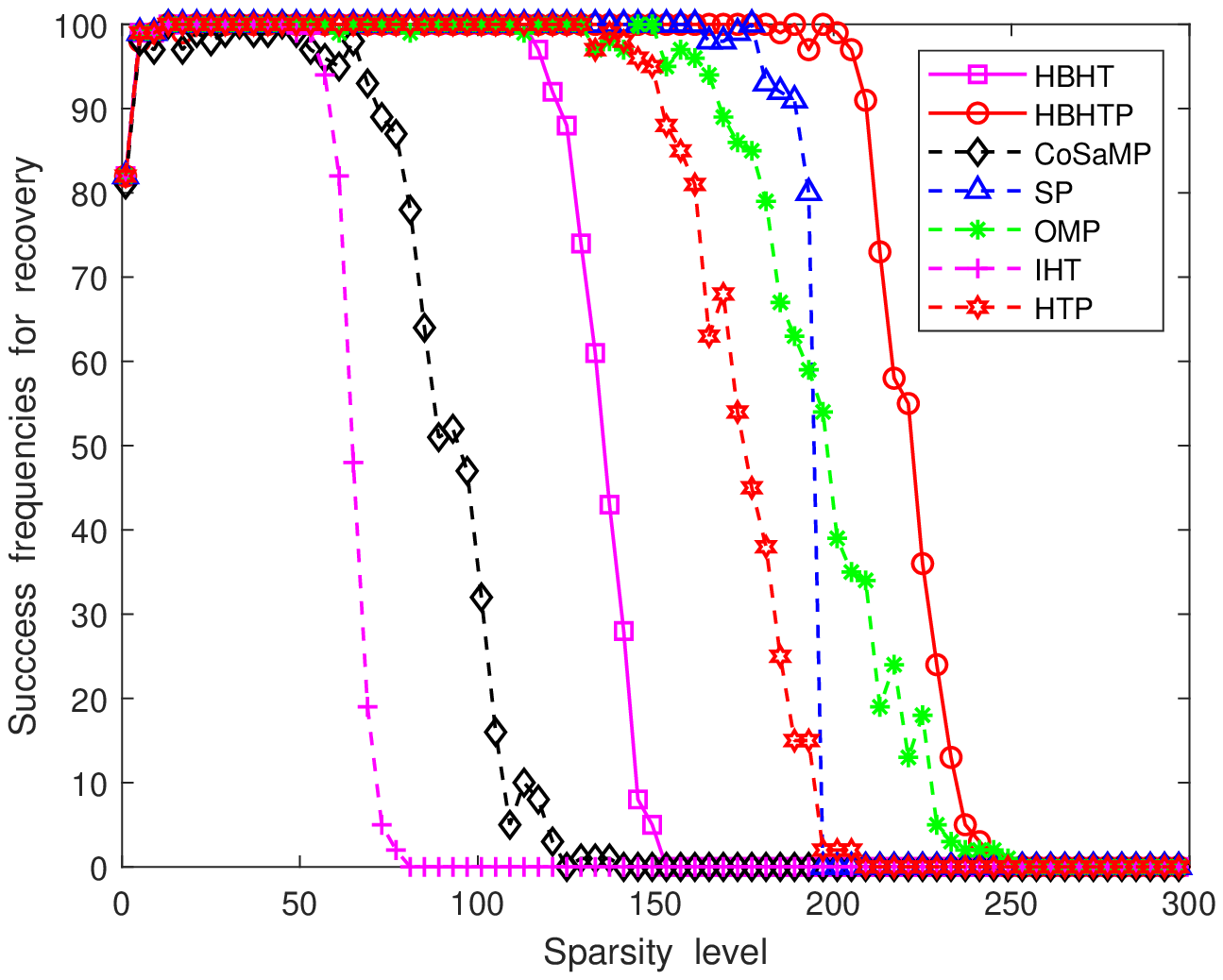}
  \end{minipage}
  }
  \caption{Comparison of success frequencies (rates) of  algorithms for signal recovery  with accurate  and inaccurate measurements, respectively. The parameters $\alpha=1.5\times 10^{-3}$ and $ \beta=0.6$ are set for HBHT and $\alpha=7\times10^{-3}$ and $ \beta=0.7$ for HBHTP. }\label{fig-comparison-OSCIH}
  \end{figure}

\subsubsection{Performance with unnormalized matrices}\label{subsect1}

The performance of iterative-type thresholding methods is closely related to the  choice of stepsize in each step.  When $A$ is un-normalized/unscaled, initial simulations indicate that $\alpha=10^{-3}$ is a proper choice for IHT and HTP,  $\alpha\in[10^{-3},2\times 10^{-3}]$ and $  \beta\in(0,0.6]$ are proper choices for HBHT, and $\alpha\in[10^{-3}, 8\times10^{-3}]$ and $ \beta\in(0,0.7]$ are suitable for HBHTP, where the range for $\beta$ is implied from (\ref{rem-parameter-1}). We now start to compare algorithms using both accurate and inaccurate measurements. Given a random pair of $(A,x^*),$  the accurate and inaccurate measurements are given respectively by $y:=Ax^*$ and $y:=Ax^*+\epsilon h$,  where $\epsilon=0.008$ and $h$ is a standard Gaussian random noise vector. We use the fixed parameters $\alpha=1.5\times 10^{-3}$ and $ \beta = 0.6$ in HBHT  and  $\alpha=7\times 10^{-3}$ and $ \beta =0.7$ in HBHTP.  The estimated success rates of the algorithms  are shown  in Fig. \ref{fig-comparison-OSCIH} in which the sparsity level $k$ is ranged from 1 to 297 with stepsize 4.  It shows that the HBHTP generally outperforms  the OMP, SP and HTP, and it might perform clearly better than  CoSaMP, HBHT and IHT. We also observe from the experiments that the success rate of HBHT is slightly worse than that of CoSaMP in noiseless settings but it might be better than CoSaMP in noisy settings.

   \begin{figure}[htbp]
   \centering
   \subfigure[Accurate measurements]{
   \begin{minipage}[t]{0.45\linewidth}
   \centering
   \includegraphics[width=\textwidth,height=0.8\textwidth]{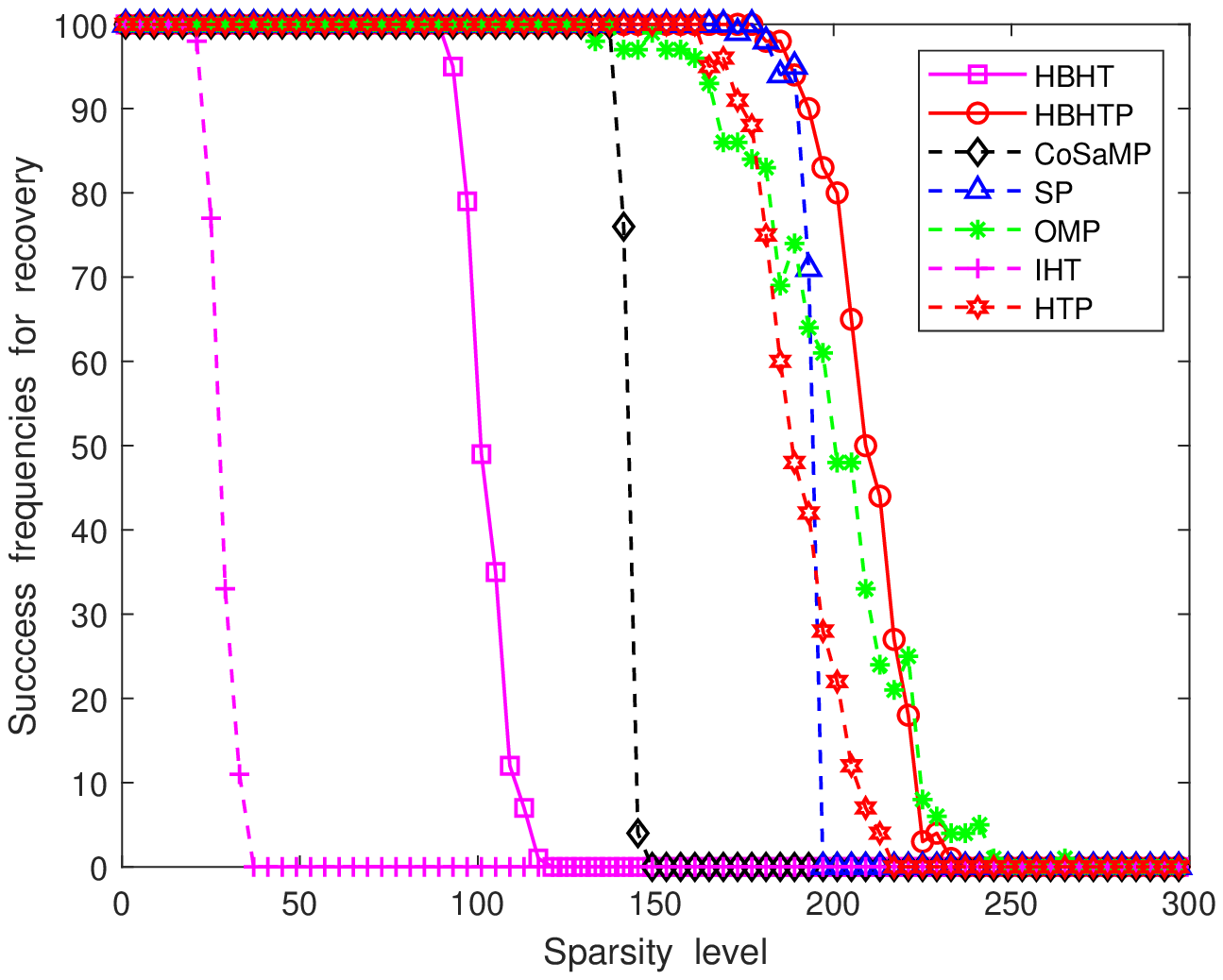}
   \end{minipage}
   }
   \subfigure[Inaccurate measurements.]{
   \begin{minipage}[t]{0.45\linewidth}
   \centering
   \includegraphics[width=\textwidth,height=0.8\textwidth]{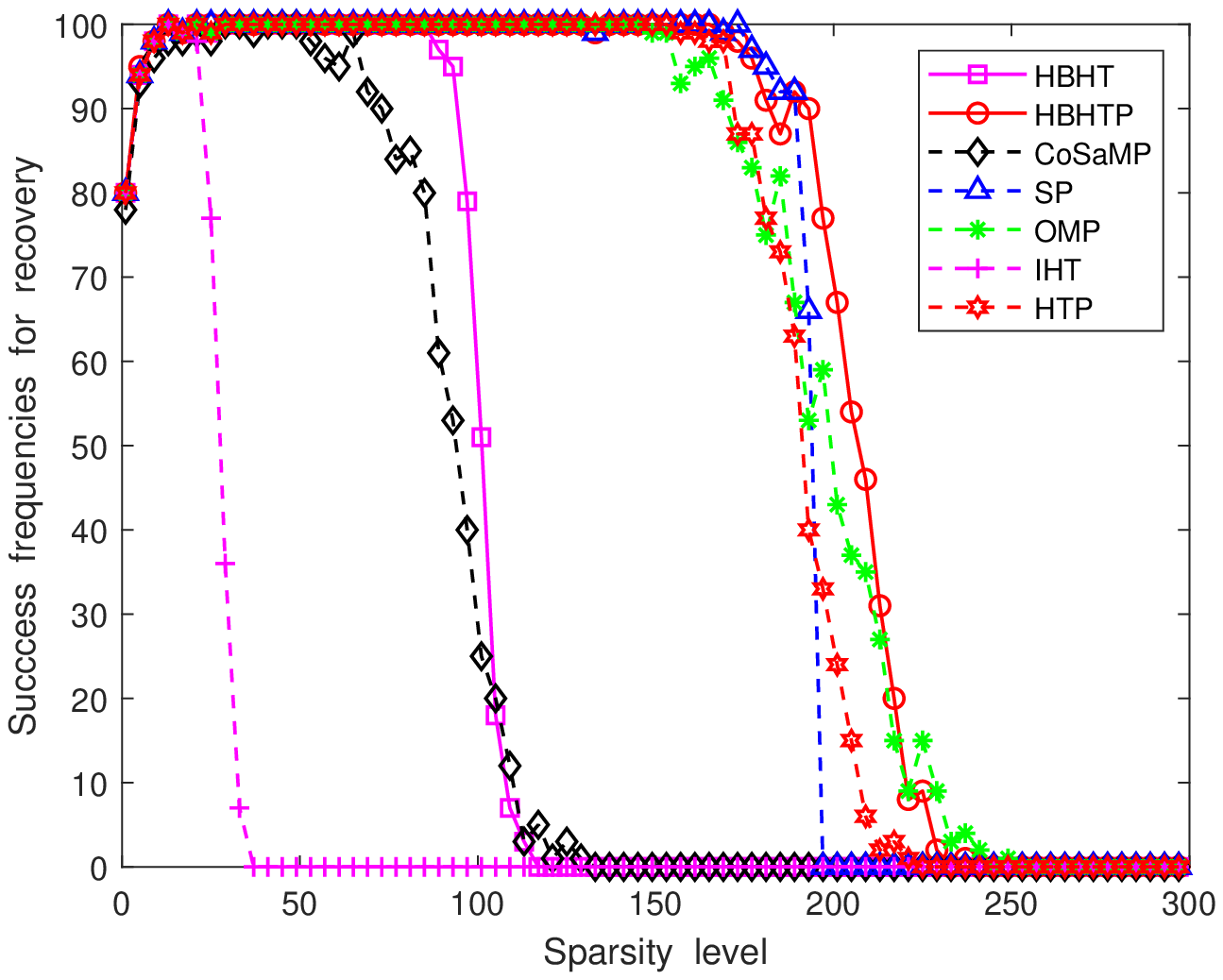}
   \end{minipage}
   }
   \caption{Comparison of success frequencies of algorithms with accurate and inaccurate measurements, respectively. The parameters $\alpha=0.6$ and $ \beta=0.1$  are set in HBHT, and $\alpha=1.7$ and $\beta=0.7$ are set in HBHTP. }\label{fig-precond-comparison-OSCIH}
   \end{figure}

\subsubsection{Performance with normalized matrices}\label{NE-precondition}

 The existing theory claims that the IHT and HTP with a larger stepsize  such as $ \alpha =1$ remains convergent if the matrix satisfies the RIP, and it is well known  that the normalized Gaussian matrix $\bar{A}:=\frac{1}{\sqrt{m}}A$ may satisfy the RIP in high probability (see, e.g., Chapter 9 in \cite{Foucart-2013} for details).
In terms of a normalized matrix, the problem (\ref{main-optimal-eq}) is equivalent to
$
\underset{z}{\textrm{argmin}}\{ {\left\lVert \bar{y}-\bar{A}z\right\rVert}_2^2: \left\lVert z\right\rVert_0\leq k\},
$
where $\bar{y}=\frac{1}{\sqrt{m}}y$. The entries of such a normalized Gaussian matrix follow the distribution $\mathcal N(0,m^{-1})$. By taking into account the theoretical results in previous sections and testing for the values of parameters  $ (\alpha, \beta),$ we found the choices $\alpha\in[0.4+2\beta,1.6] $ and $ \beta\in(0,0.6]$ are suitable for HBHT and
$\alpha\in[0.3+2\beta,1.7]$ and $ \beta\in(0,0.7]$ are suitable for HBHTP to achieve a good performance.  We repeated the experiments in  Section \ref{subsect1} by setting the stepsize $ \alpha=1 $ for IHT and HTP, the specific values $\alpha=0.6$ and $\beta=0.1$ for HBHT and $\alpha=1.7$ and $ \beta=0.7$  for HBHTP. The results are demonstrated in Fig. \ref{fig-precond-comparison-OSCIH} which appear to be similar to that of Fig. \ref{fig-comparison-OSCIH}.  However, one can observe that the normalization of the matrix, accordingly enlarged stepsize, and the choices of parameters do affect the recovery ability of HBHT, HBHTP, IHT and HTP to a certain degree.  Again, it seems that the HBHTP performs generally better than other algorithms in noiseless and noisy settings, and the HBHT  may perform better than CoSaMP and IHT in some noisy situations. Compared to IHT and HTP, the heavy-ball-based technique does play a vital role in speeding up and enhancing the performance of the traditional thresholding algorithms for sparse signal recovery.

\subsubsection{Average number of iterations and time}

We now compare the average number of iterations and CPU time required by several algorithms to meet the criterion (\ref{recov-criter}) with accurate measurements. The testing environment is the same as Section \ref{subsect1}.  Within 50 iterations, if $x^p$ satisfies criterion  \eqref{recov-criter}, the algorithm
 terminates and the number of iterations $p$ is recorded. Otherwise, the number of iterations is recorded as 50. For OMP, the number of iterations is equal to the sparsity level of the input signal.   Fig. \ref{fig-precond-time-numiter-accu}(a) indicates that the average number of iterations required by HBHTP  is lower than that of  SP and HTP, and might be  much lower than that of OMP, CoSaMP, HBHT  and IHT especially when the sparsity level $k$ is high.

 \begin{figure}[htbp]
  \centering
  \subfigure[Average number of iterations]{
  \begin{minipage}[t]{0.45\linewidth}
  \centering
  \includegraphics[width=\textwidth,height=0.8\textwidth]{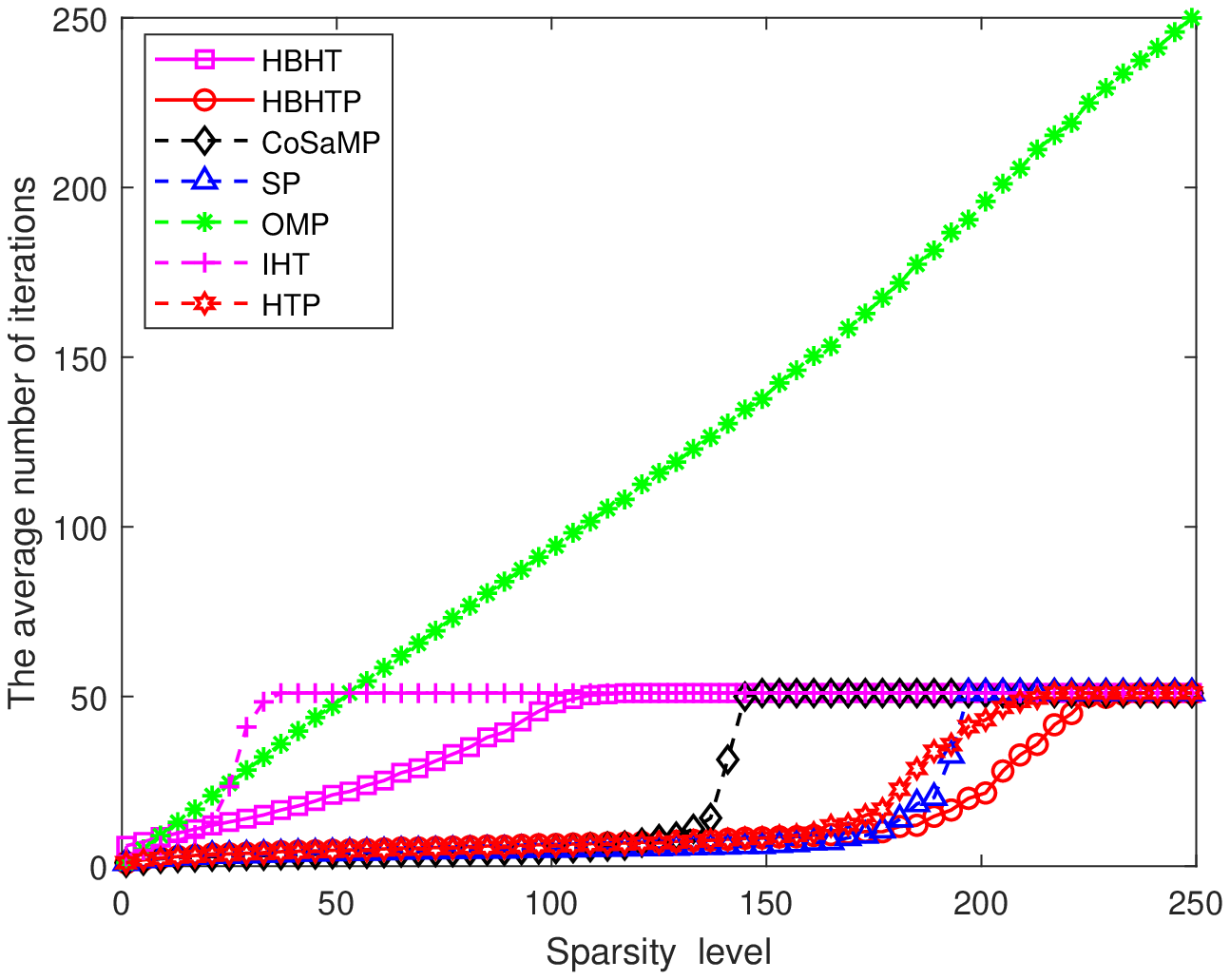}
  \end{minipage}
  }
  \subfigure[Average CPU time]{
  \begin{minipage}[t]{0.45\linewidth}
  \centering
  \includegraphics[width=\textwidth,height=0.8\textwidth]{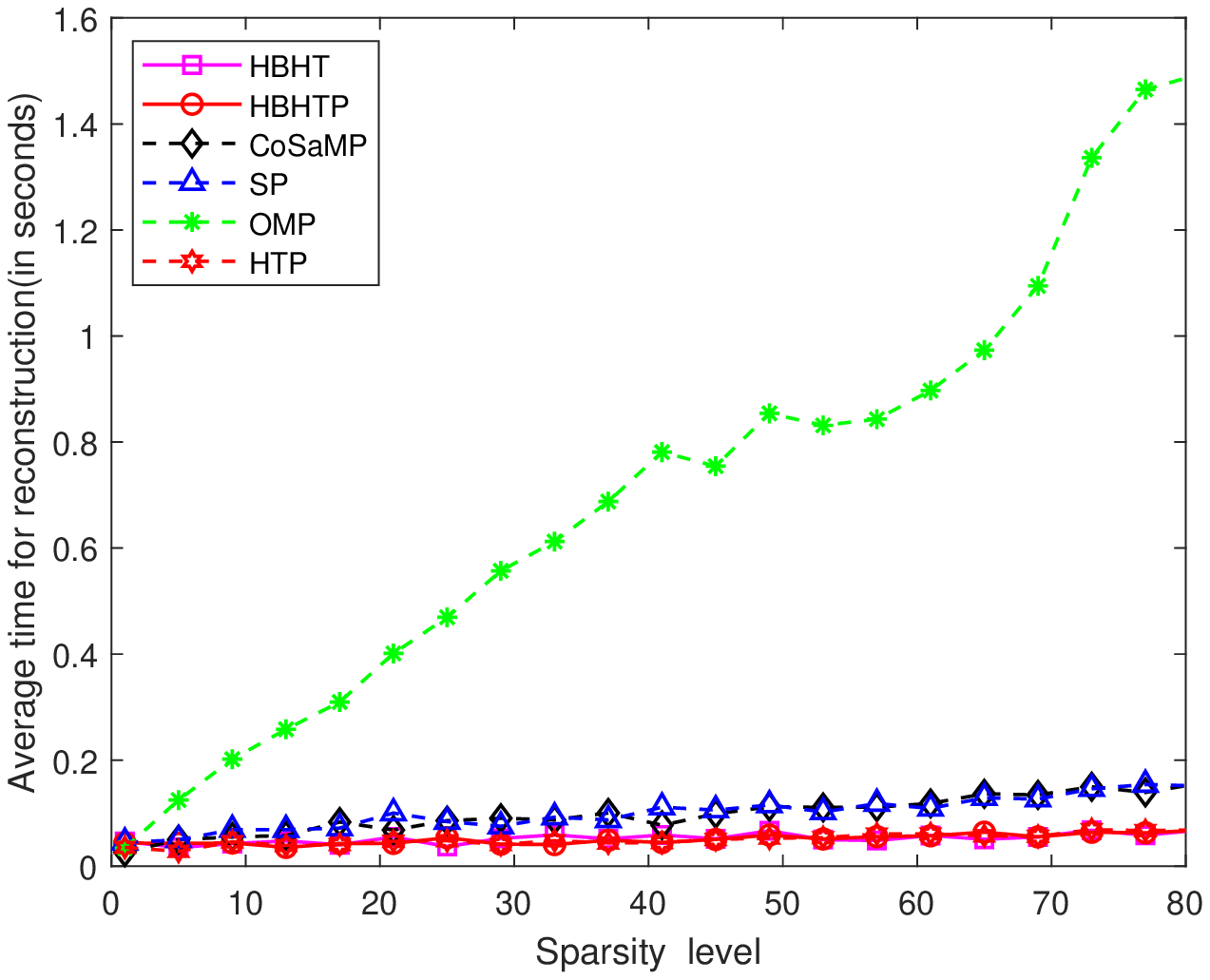}
  \end{minipage}
  }
  \caption{Comparison of  average number of iterations and time taken by algorithms to meet the recovery  criterion (\ref{recov-criter}) with accurate measurements. The parameters  $\alpha=0.6 $ and $ \beta=0.1$ are set in HBHT and $\alpha=1.7$  and $ \beta=0.7$ in HBHTP. }\label{fig-precond-time-numiter-accu}
  \end{figure}

 Fig. \ref{fig-precond-comparison-OSCIH}(a) indicates that   all $k$-sparse signals with  $k\leq 80$ can be recovered by all mentioned algorithms except IHT. Thus we focus on the signals with sparsity levels $k\leq 80$ to compare the average time consumed by algorithms  except IHT to meet the criterion \eqref{recov-criter}. The results are demonstrated in Fig. \ref{fig-precond-time-numiter-accu}(b), from which one can see that OMP takes more time  than other algorithms to recovery the signal, and that  the average time taken by SP, CoSaMP, HBHTP, HBHT and HTP increases slowly in a linear manner with respect to the sparsity level $k,$  and the average time consumed by SP and CoSaMP is approximately twice of HBHT, HBHTP  and HTP.  This indicates that the proposed algorithms have some advantage in time saving for signal recovery.

  \subsection{Phase transition}\label{NE-phase-transition}

We further investigate and compare the performances of algorithms through the empirical phase transition curves (PTC) and algorithm selection maps (ASM) introduced in  \cite{blanchard2015,blanchard2015cgiht}.  All $m\times n$ matrices in this subsection are Gaussian random matrices with fixed $n=2^{12}$, whose  entries are iid and follow the distribution $\mathcal N(0,m^{-1})$. The parameters $(\alpha,\beta)$ in HBHT and HBHTP are set exactly the same as in Section \ref{NE-precondition}.

 \subsubsection{Phase transition curves}\label{phase-transition-curves}
Denote by $\delta=m/n$ and $\rho=k/m.$  The phase transition curve of an algorithm separates the $(\delta, \rho)$ space into \emph{success} and \emph{failure} regions. The region below the curve, called recovery region, represents the problem instances with $(\delta, \rho)$ that can be exactly or approximately solved by the algorithm, while the region above  the curve indicates the problem instances with $ (\delta, \rho)
$ to which the algorithm does not appear to find their correct solutions. The empirical phase transition curves demonstrated in this section are logistic regress curves identifying the 50\% success rate for the given algorithm applying to a given  problem class. This method was first introduced in \cite{blanchard2015, blanchard2015cgiht}.

\begin{figure}[htbp]
   \centering
  \subfigure[Accurate  measurements]{
   \begin{minipage}[t]{0.3\linewidth}
   \centering
  \includegraphics[width=\textwidth,height=0.9\textwidth]{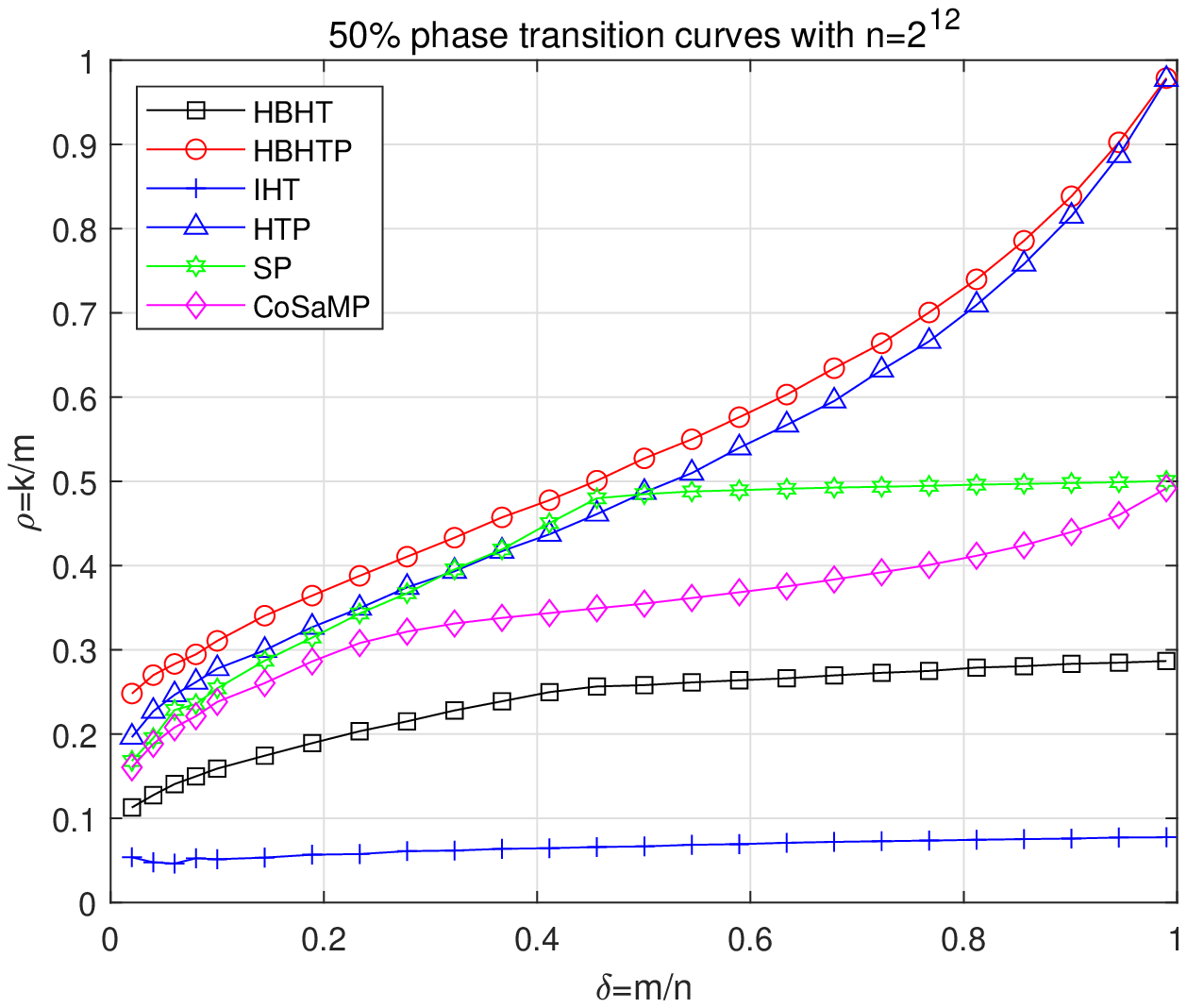}
  \end{minipage}
    }
   \subfigure[Inaccurate  measurements]{
   \begin{minipage}[t]{0.3\linewidth}
  \centering
 \includegraphics[width=\textwidth,height=0.9\textwidth]{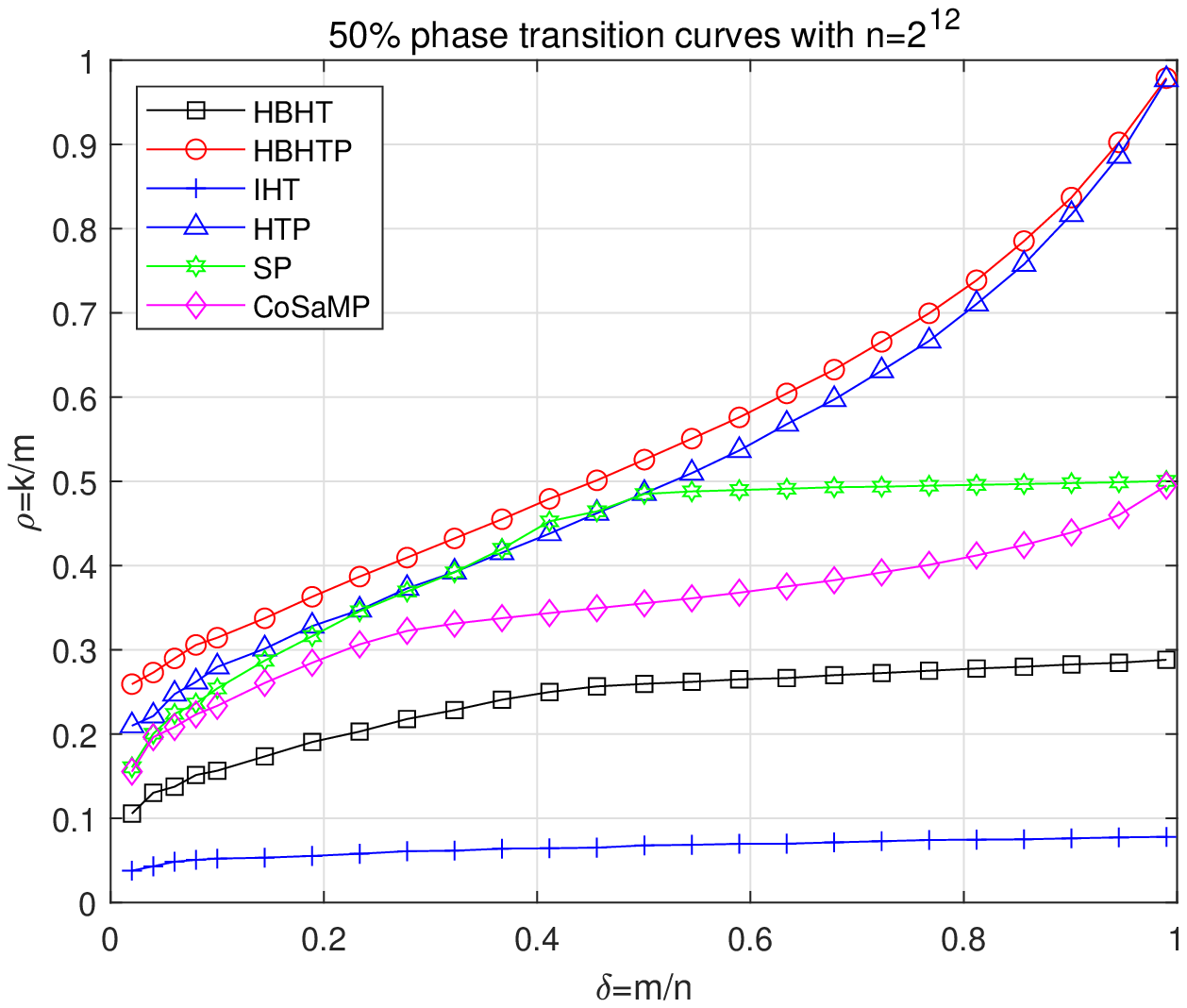}
    \end{minipage}
    }
    \subfigure[Different noises]{
      \begin{minipage}[t]{0.3\linewidth}
      \centering
     \includegraphics[width=\textwidth,height=0.9\textwidth]{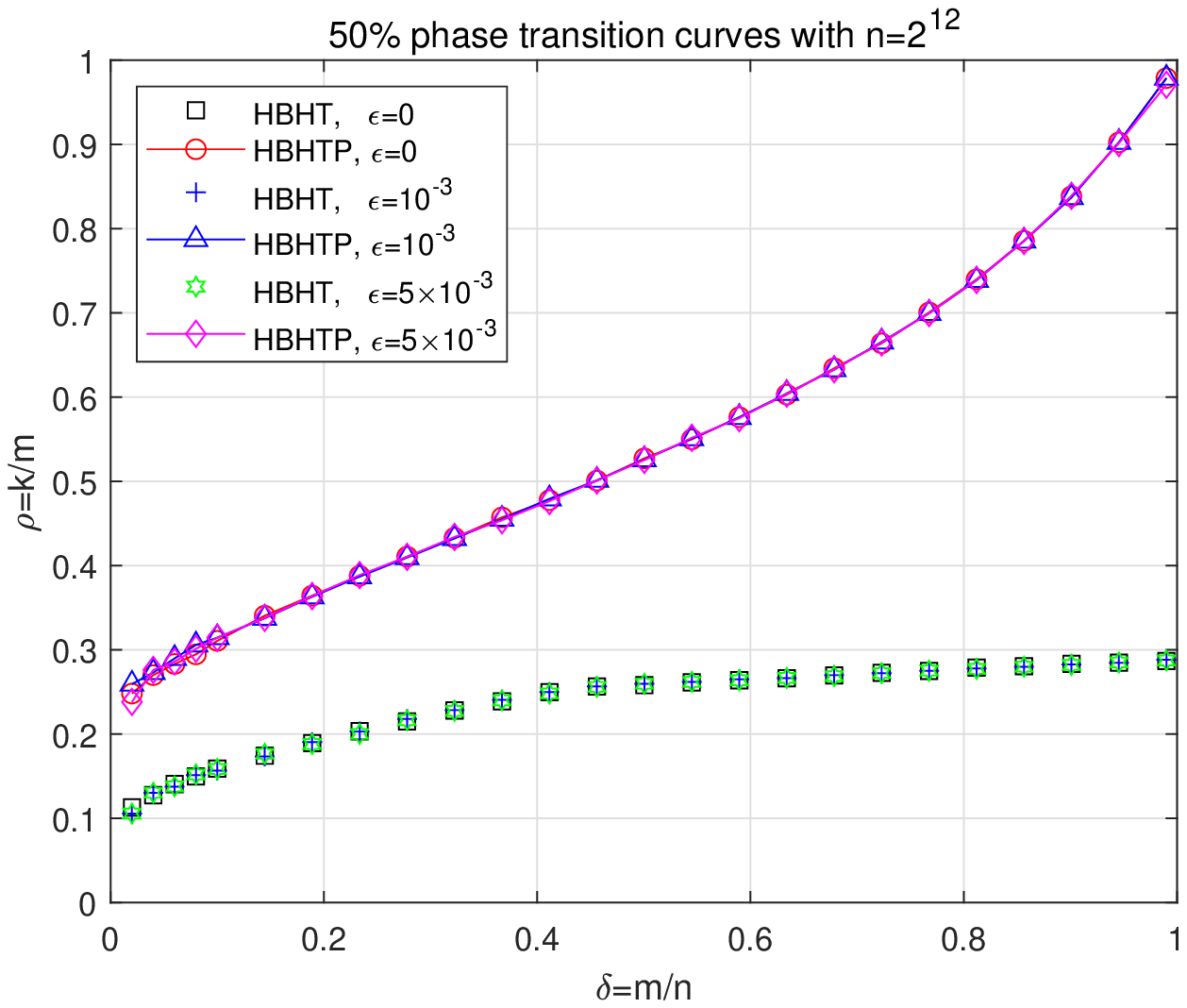}
     \end{minipage}
       }
   \caption{The $50 \%$ success rate phase transition curves  for six algorithms.  }\label{fig-RPTC-six-alg}
   \end{figure}
 We now briefly introduce the mechanism for generating such a curve. The interested readers may find more detailed information about this from the references \cite{blanchard2015, blanchard2015cgiht}. To generate the PTC and ASM, we consider 25 different values of  $m=\lceil \delta\cdot n\rceil$ where
  \begin{equation}\label{delta-value}
\delta\in\{0.02,0.04,0.06,0.08\}\cup\{0.1,0.1445,\ldots,0.99\},
  \end{equation}
where the interval $[0.1,0.99]$ was equally divided into 20 parts. For every value of  $m$, we collect 50 groups of sparsity levels $k=\lceil\rho\cdot m \rceil$ where $\rho$  is ranged from 0.02 to 1 with stepsize 0.02. For a fixed $m$, the recovery phase transition region for each algorithm is estimated by the interval $[k_{\min},k_{\max}]$, where $k_{\min}$ and $k_{\max}$ can be determined by a bisection method. They are  the critical values to ensure that the recovery success rate is at least 90\%   for any $k<k_{\min}$ and at most 10\% for any $k>k_{\max}$. For simplicity, we introduce the notations
$k_j \stackrel{\Delta}{=} k_{\min}+\lceil j\cdot \Delta k\rceil(j=0,1,\ldots,J)$, where $ \Delta k=(k_{\max}-k_{\min})/J$ and $J=k_{\max}-k_{\min} $ if $k_{\max}-k_{\min}<50;$ otherwise $ J=50. $
When estimating the success rate of an algorithm, $Nb=10$ problem instances are tested for each given $(k,m,n)$, where $k=k_j, j=0,1,\ldots,J$. Based on the success rates, the phase transition curves can be obtained from the following logistic regression model \cite{blanchard2015,blanchard2015cgiht}:
$$
\min_{(\gamma_0,\gamma_1)}\sum_{j=0} ^{J}\left|g(k_j/m)-\frac{suc(k_j,m,n)}{Nb}\right|,
$$
where
$$
g(\rho)=\frac{1}{1+exp[-\gamma_0(1-\gamma_1\rho)]},
$$
 and $suc(k_j,m,n)$ is the number of recovery success among $Nb$ problem instances for each $(k_j,m,n), j=0,1,\ldots,J.$  The 50\% success recovery phase transition curves are defined by $g(\rho)=0.5.$

  The curves  for the algorithms HBHT, HBHTP, IHT, HTP, CoSaMP and SP are summarized in Fig. \ref{fig-RPTC-six-alg}. In this comparison, the parameters  $\alpha=0.6$ and $\beta=0.1$ are used in HBHT   and $\alpha=1.7$ and $\beta
  =0.7$  in HBHTP. The accurate and inaccurate measurements are given by $y=Ax^*$ and $y=Ax^*+\epsilon h$, respectively, where $h$ is a standard Gaussian random vector  and $\epsilon=0.001$. From Fig. \ref{fig-RPTC-six-alg} (a) and (b), we see that HBHTP  has the highest  phase transition curves. This indicates that HBHTP may outperform the other five algorithms for sparse signal recovery in both noiseless and noisy environments. One can also see that  the phase transition curves  of SP, CoSaMP, HBHT and IHT are below the line  $\rho=0.5$ as $\delta\geq 0.5.$  This implies that the recovery performance of these algorithms would not remarkably be improved even when the number of measurements  is increased. By contrast, the phase transition curves of HBHTP and HTP  are twice as high as those of SP and CoSaMP as $\delta\rightarrow 1$. To see the influence of noise levels on the performance of algorithms, the phase transition curves for HBHT and HBHTP with three different noise levels $\epsilon\in\{0,10^{-3},5\times 10^{-3}\}$ are demonstrated in Fig. \ref{fig-RPTC-six-alg}(c), from which one can observe that the curves of HBHT and HBHTP do not significantly change with respect to the noise level when the noise level is relatively low. This sheds light on the stability of the two algorithms in signal recovery.

 \begin{figure}[htbp]
            \centering
            \begin{minipage}[t]{0.45\linewidth}
            \centering
            \includegraphics[width=\textwidth,height=0.8\textwidth]{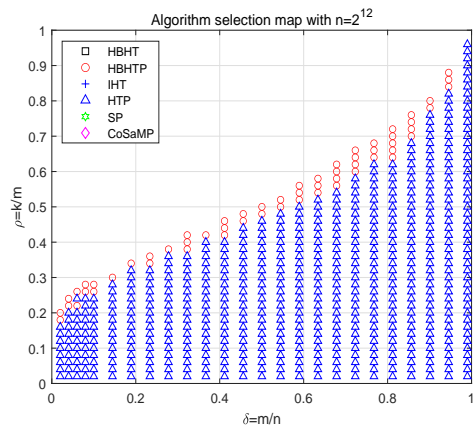}
            \end{minipage}
            \caption{Selection map with accurate  measurements.  }\label{fig-Selection-map}
            \end{figure}

    \subsubsection{Algorithm selection map}\label{Algorithm-selcetions}

   The intersection of the recovery regions below the phase transition curves indicates that multiple algorithms are capable of signal recovery. To choose an algorithm,
   one might also consider the computational time for recovery. As a result, the so-called algorithm selection map was introduced in \cite{blanchard2015,blanchard2015cgiht}, which demonstrates the least average recovery time of the algorithms with accurate measurements. To draw an algorithm selection map,  for each $\delta$ taking the values in \eqref{delta-value}, 10 problem instances are tested for every algorithm on the sampled phase space with the mesh  $(\delta,\rho)$ with
   $\rho=\{j/50, j=1,2,\ldots,50\}$ until the success rate is lower than 90\%. The algorithm with least computational time  will be identified on the map. The map is shown in Fig. \ref{fig-Selection-map}, which clearly depicts two regions in the phase plane,  wherein HBHTP is the fastest algorithm for solving problem instances with relatively large $\rho,$ while the HTP reliably recovers the signal in least time in other cases.

\begin{figure}[htbp]
               \centering
                 \subfigure[Average time for the fastest algorithm]{
                        \begin{minipage}[t]{0.3\linewidth}
                        \centering
                        \includegraphics[width=\textwidth,height=0.8\textwidth]{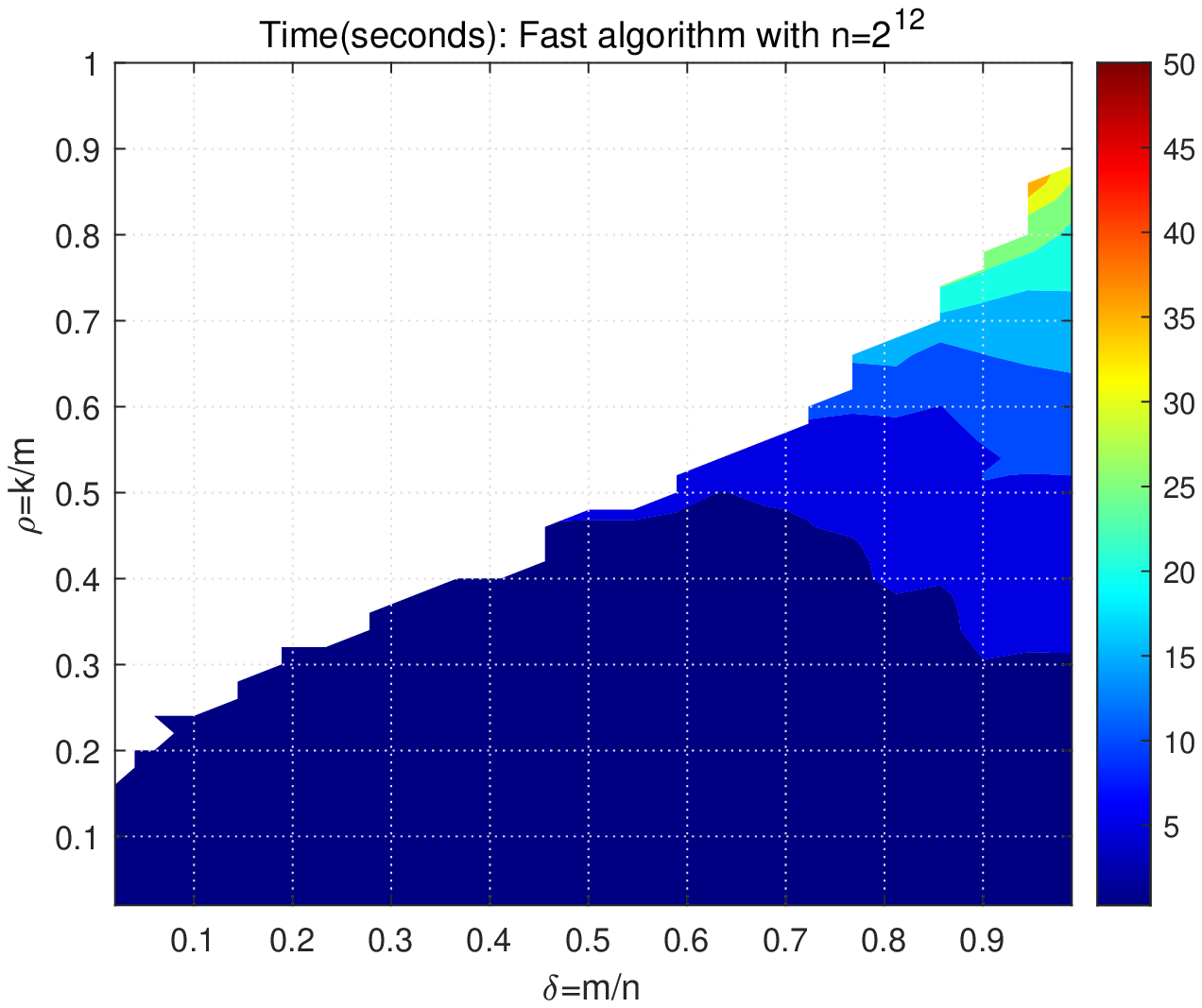}
                        \end{minipage}
                        }
               \subfigure[HBHTP]{
               \begin{minipage}[t]{0.3\linewidth}
               \centering
               \includegraphics[width=\textwidth,height=0.8\textwidth]{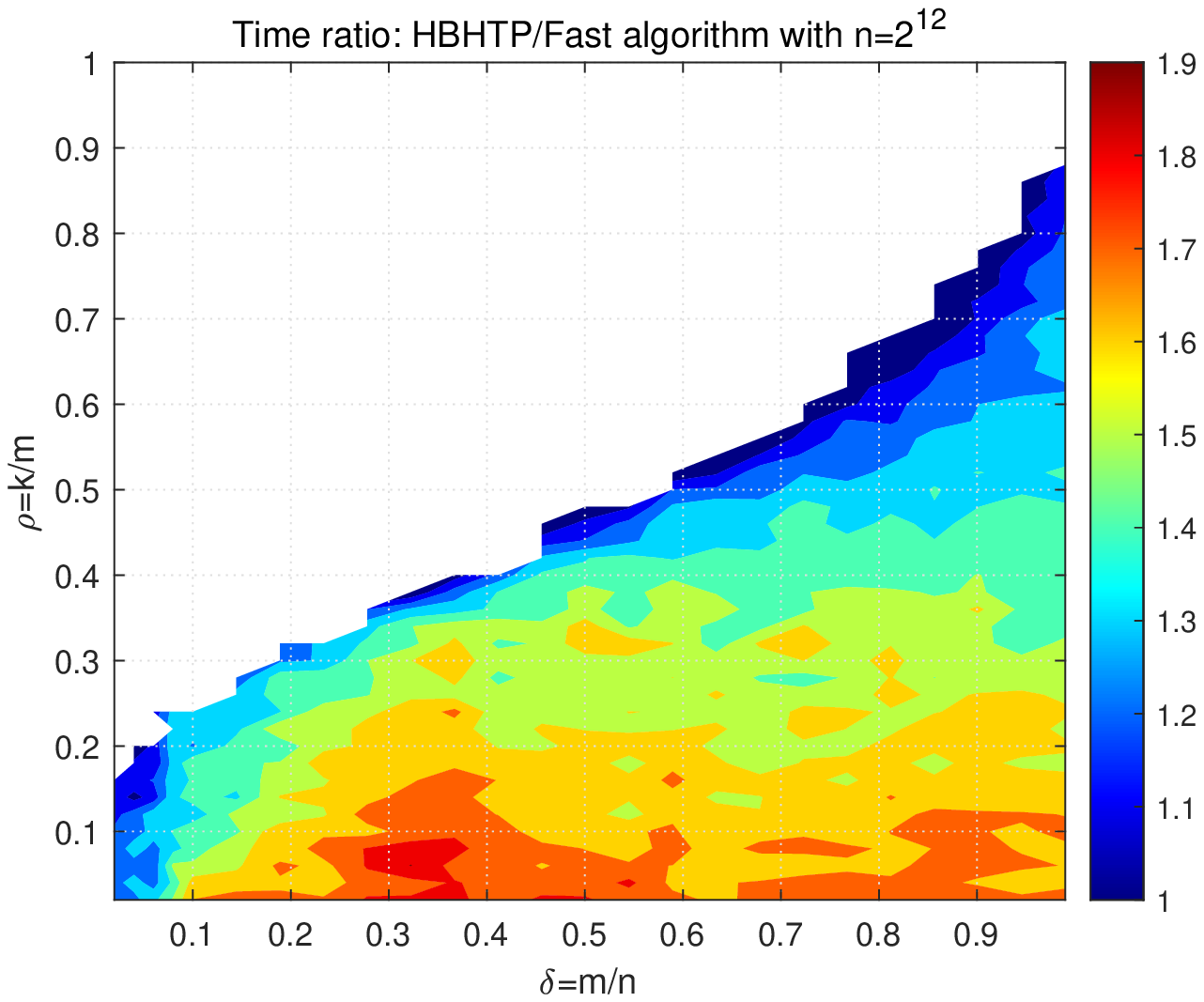}
               \end{minipage}
               }
                 \subfigure[HBHT]{
                              \begin{minipage}[t]{0.3\linewidth}
                              \centering
                              \includegraphics[width=\textwidth,height=0.8\textwidth]{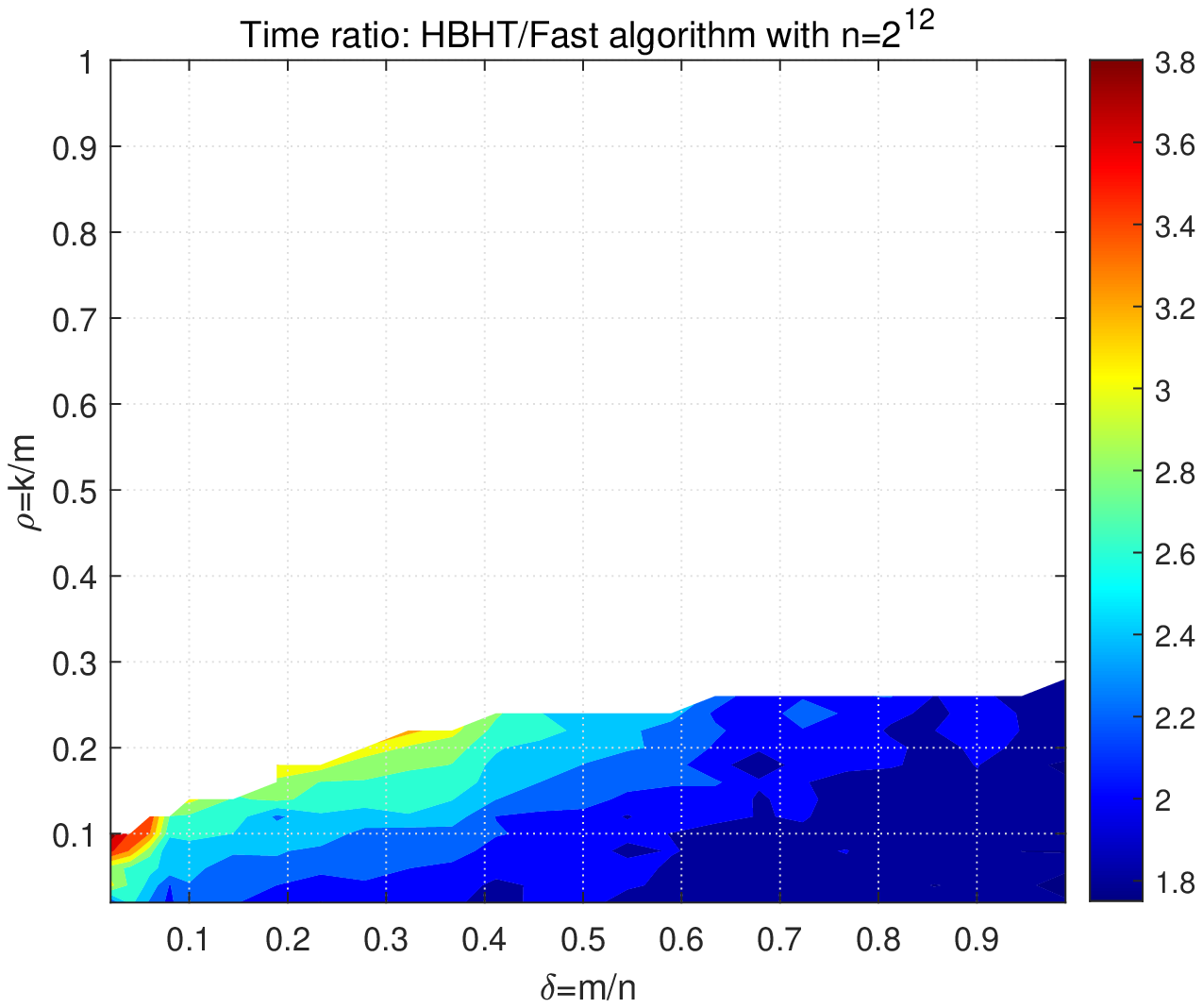}
                              \end{minipage}
                              }
              \subfigure[HTP]{
                \begin{minipage}[t]{0.3\linewidth}
                \centering
                \includegraphics[width=\textwidth,height=0.8\textwidth]{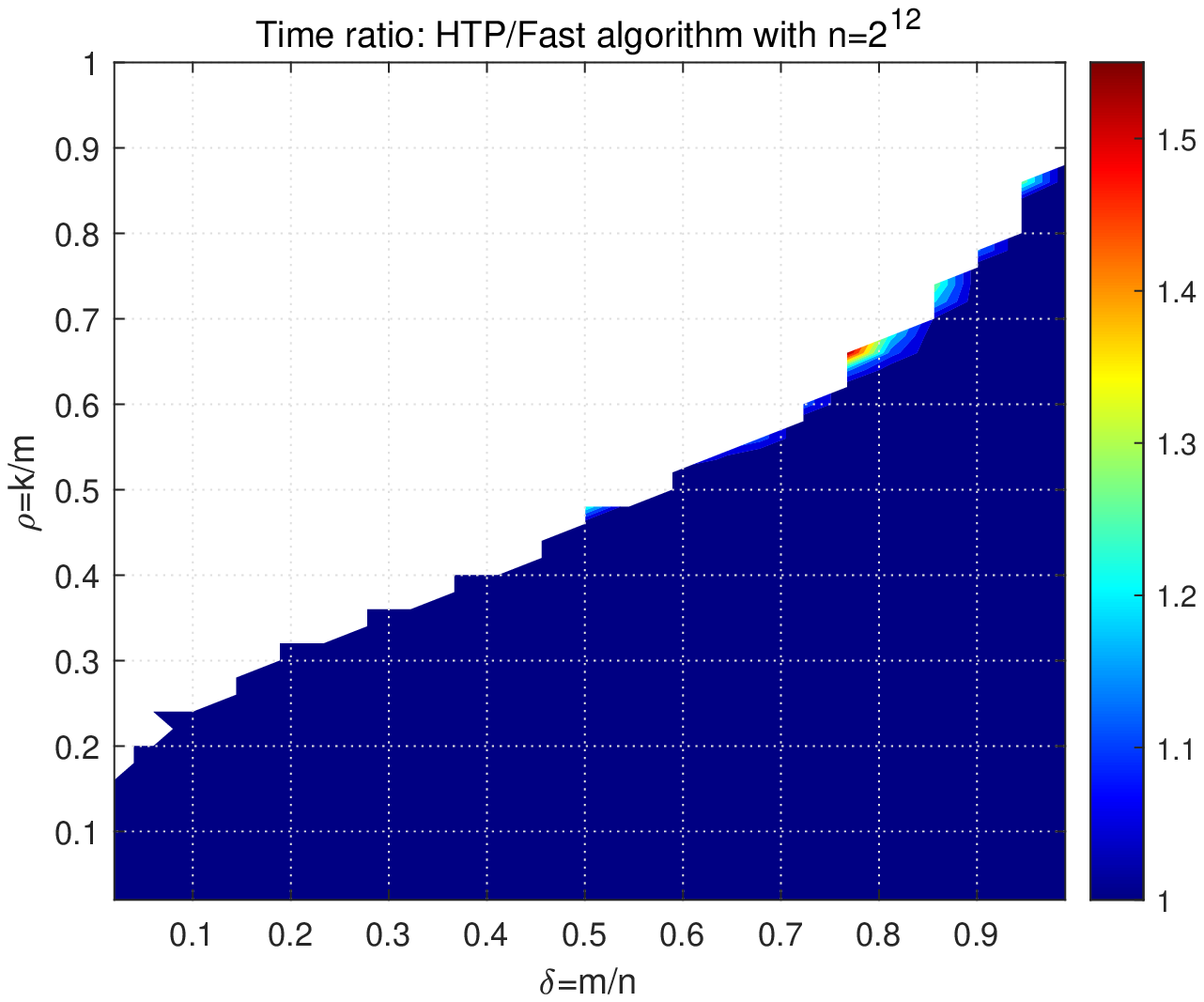}
                \end{minipage}
                }
             \subfigure[SP]{
                \begin{minipage}[t]{0.3\linewidth}
                \centering
                \includegraphics[width=\textwidth,height=0.8\textwidth]{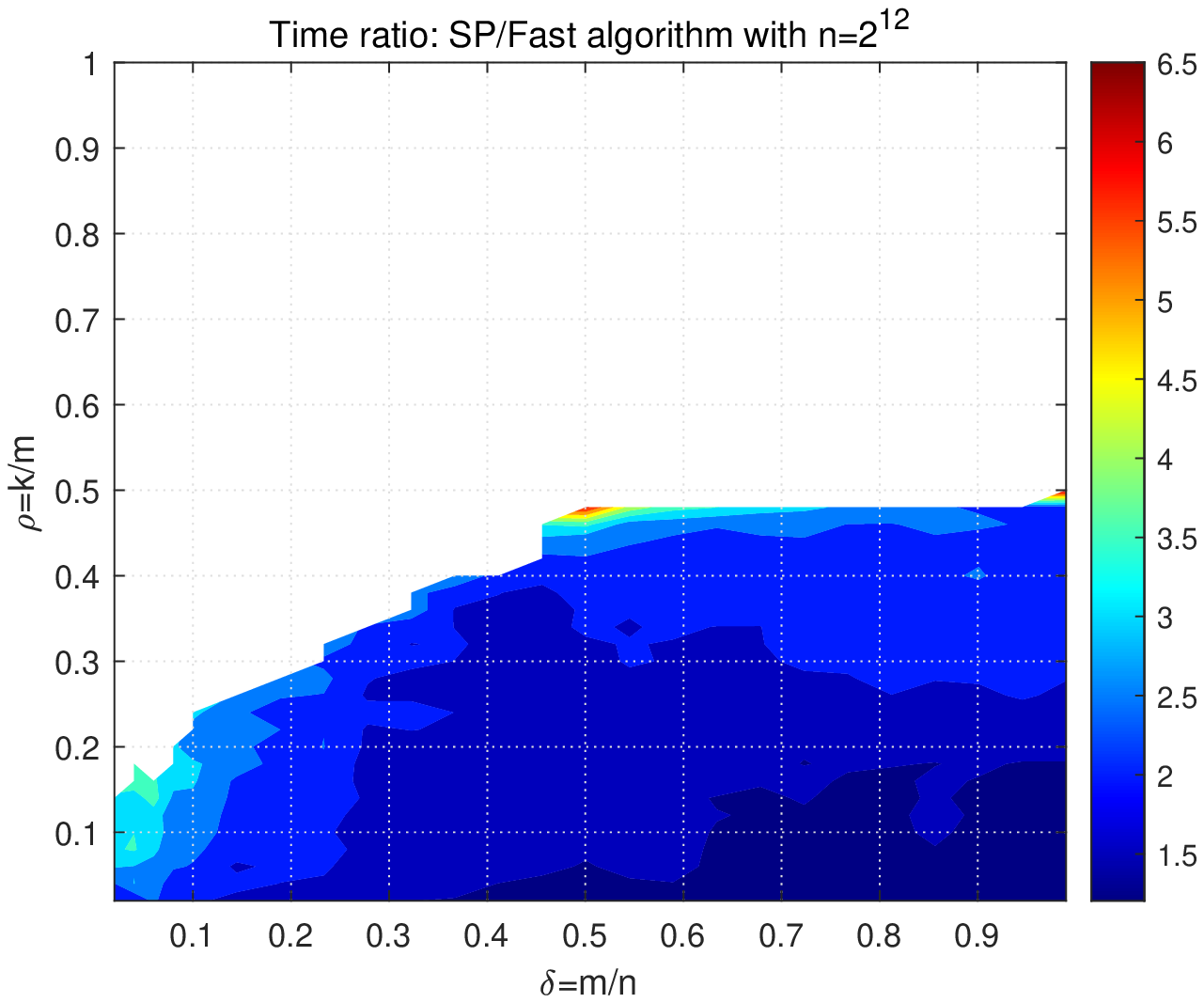}
                \end{minipage}
                }
              \subfigure[CoSaMP]{
                \begin{minipage}[t]{0.3\linewidth}
                \centering
                \includegraphics[width=\textwidth,height=0.8\textwidth]{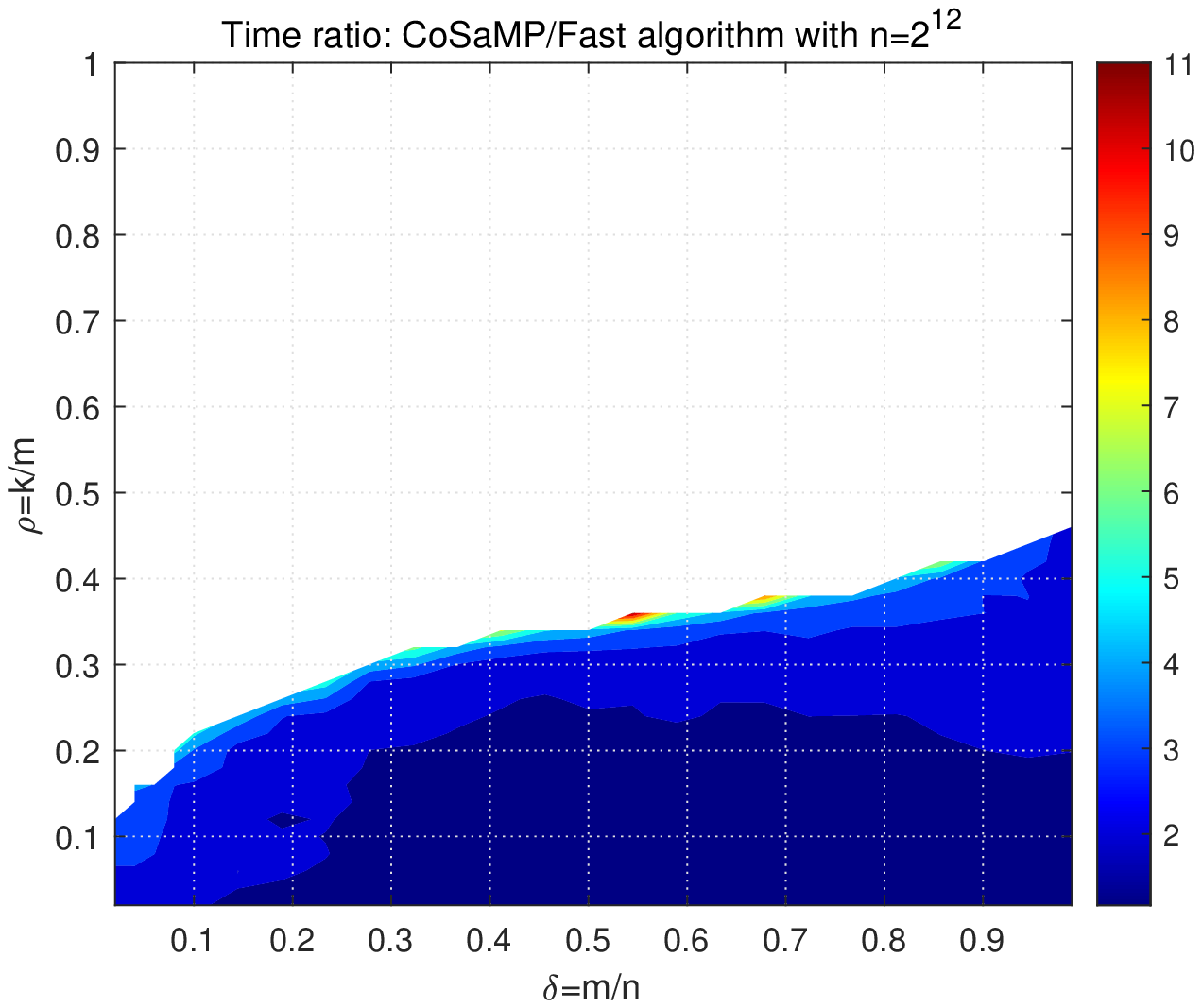}
                \end{minipage}
                }
               \caption{ (a) The minimum average time of algorithms; (b)-(f) The ratios of average time for several algorithms against the fastest one.  }\label{fig-time-ratio}
               \end{figure}
After identifying the fastest algorithm, further information on the  average recover time of algorithms are  given in Fig.  \ref{fig-time-ratio}. The minimum average recovery time  taken by an algorithm is displayed in Fig. \ref{fig-time-ratio}(a). When $\rho\leq0.3$, the minimal average run time  of algorithms is close to each other for any $\delta\in(0,1)$.  However, when $\rho> 0.3$, we see that the larger the value of $ \rho$, the more average run time is required by the algorithm when $\delta>0.6. $ The ratios of the average recovery time for the algorithms HBHTP, HBHT, HTP, SP  and CoSaMP against that of the fastest algorithm are displayed in Fig. \ref{fig-time-ratio}(b)-(f), respectively.   Fig. \ref{fig-time-ratio}(b) shows that the larger the value of $\rho$, the smaller the ratio for a fixed  $\delta$, and the  ratio for HBHTP is less than 1.5 when $ \rho\geq 0.25$ or $\delta\leq0.1.$ By contrast, Fig. \ref{fig-time-ratio}(c)-(f) show that the larger the value of $\rho,$ the larger the ratios for those four algorithms. This phenomenon indicates that HBHTP might work better than other algorithms when the sparsity level $k$ is relatively high. We also observe that HBHTP and HTP are comparable to each other, and that HBHT, SP and  CoSaMP often consume more than twice of the minimal average time. One can also observe that the ratios for SP and CoSaMP can be three and five times higher, respectively, when $\rho$ is large.

\begin{figure}[htbp]
          \centering
          \subfigure[$\delta=0.2780$]{
          \begin{minipage}[t]{0.3\linewidth}
          \centering
          \includegraphics[width=\textwidth,height=0.9\textwidth]{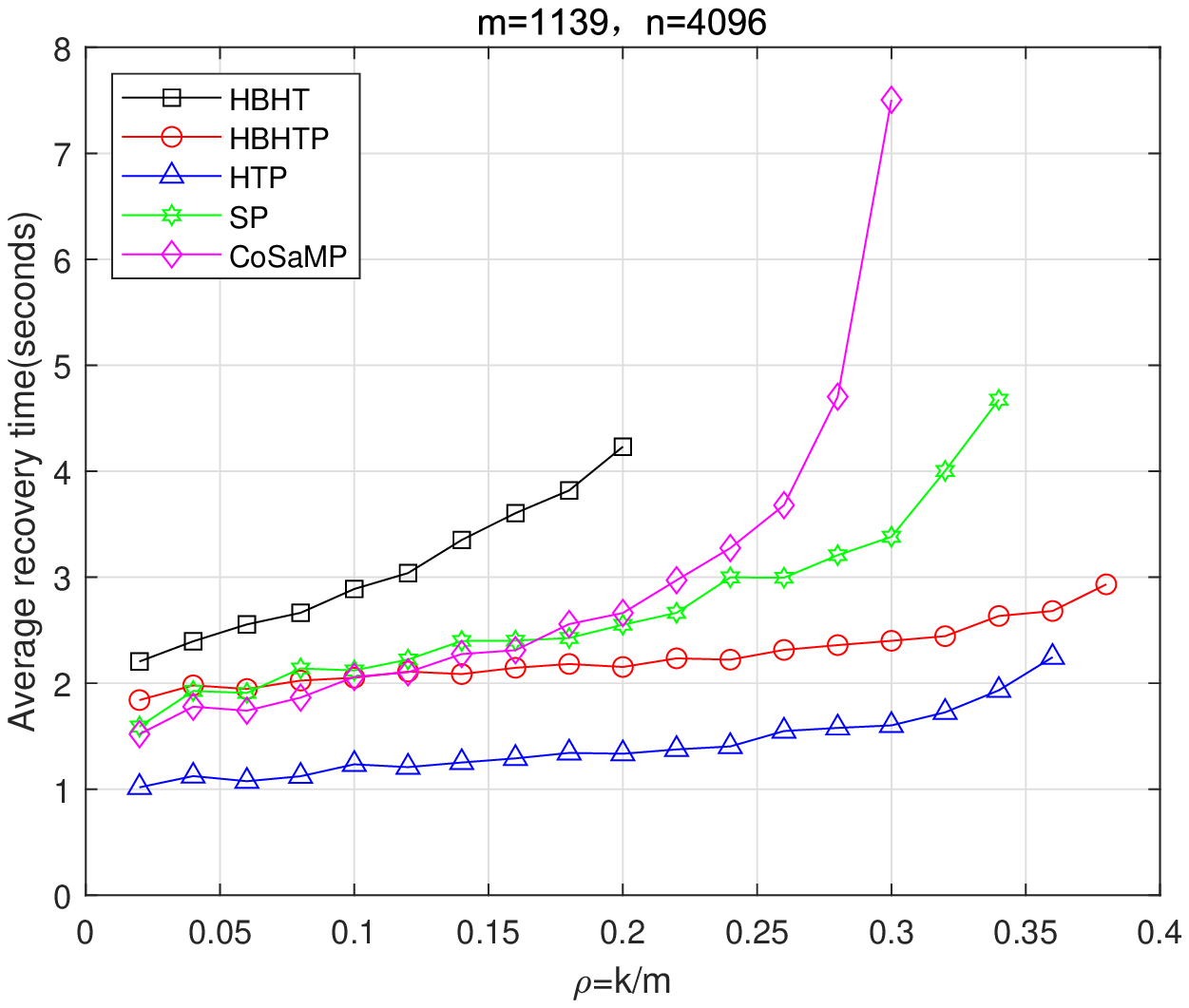}
          \end{minipage}
          }
          \subfigure[$\delta=0.5005$]{
          \begin{minipage}[t]{0.3\linewidth}
          \centering
          \includegraphics[width=\textwidth,height=0.9\textwidth]{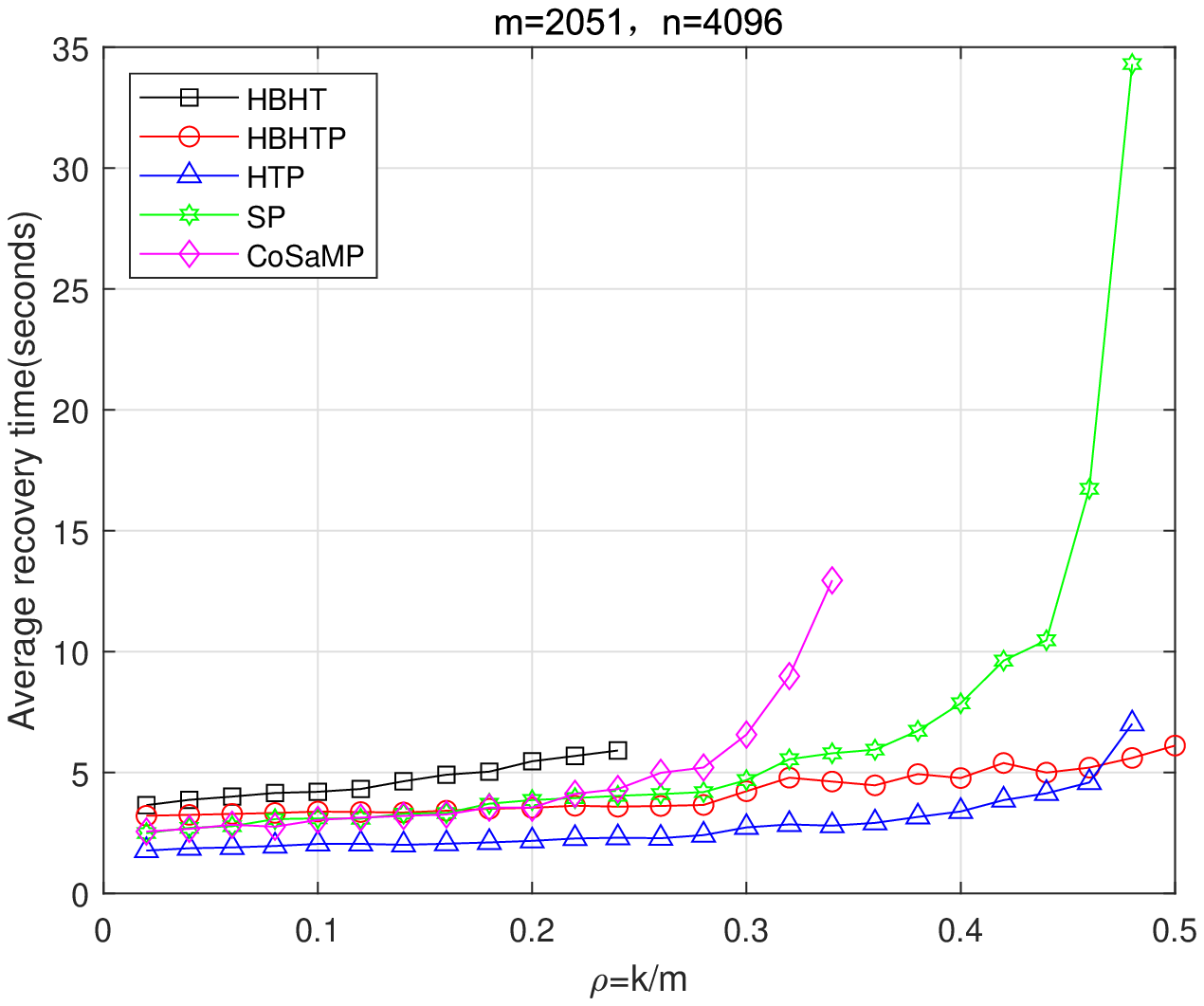}
          \end{minipage}
          }
           \subfigure[$\delta=0.7230$]{
                    \begin{minipage}[t]{0.3\linewidth}
                    \centering
                    \includegraphics[width=\textwidth,height=0.9\textwidth]{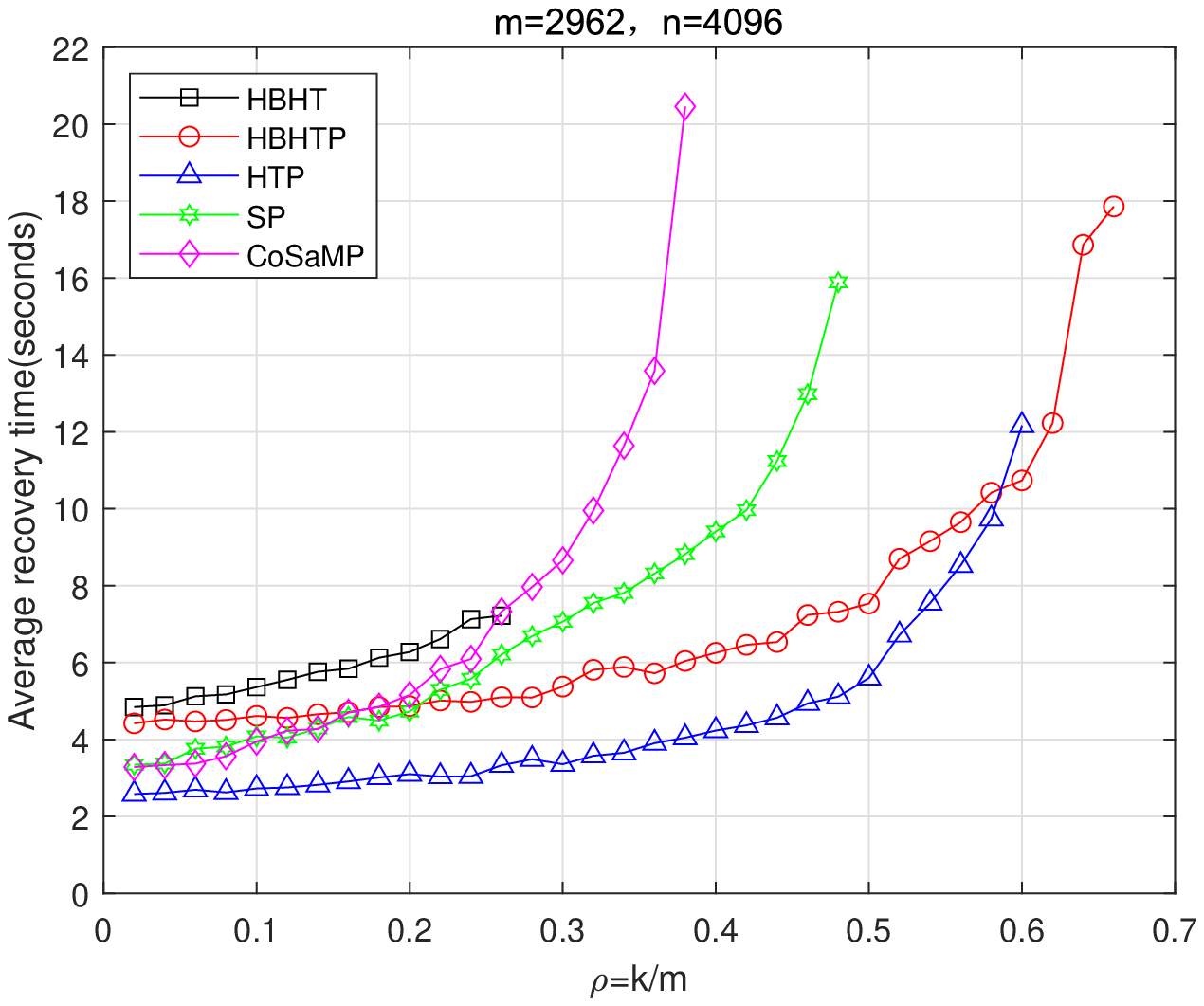}
                    \end{minipage}
                    }
          \caption{Average recovery time with respect to the change of  $\rho$ under three different fixed values of $\delta.$  }\label{fig-time-fixm}
          \end{figure}

   Finally, we demonstrate the change of average recovery time of algorithms against the factor $\rho.$ The results for three different parameters $\delta\in\{0.2780,0.5005,0.7230\}$ are given in Fig. \ref{fig-time-fixm}. For $\rho\leq0.2$, the average recover time of HBHTP, SP and CoSaMP are similar to each other.  For $\rho\in[0.2,0.5]$, the time consumed by HBHTP and HTP increases slowly compared to that of SP and CoSaMP as the sparsity level $k$ increases. Moreover, the computational time of HTP approaches and surpasses that of HBHTP for $\rho\geq0.4$ in Fig. \ref{fig-time-fixm}(b) and
    for $\rho\geq0.5$ in Fig. \ref{fig-time-fixm}(c), respectively.  Finally, we find that only HBHTP is typically able to recover the sparse signals fell into the region of the far right of Fig. \ref{fig-time-fixm}(a)-(c). This provides some evidence to show that the HBHTP might admit a certain advantage in sparse signal recovery over several existing  algorithms especially when  $\rho$ is relatively large.

\section{Conclusions}\label{conclusion}

Incorporating the heavy-ball acceleration technique into the  IHT and HTP methods leads to  the HBHT and HBHTP algorithms for sparse signal recovery. The guaranteed performance of these algorithms has been established  under the RIP assumption and certain conditions for the proper choice of the algorithmic parameters. The finite convergence and recovery stability of the algorithms were also shown in this paper.  The numerical performance of the algorithms has been investigated from several difference perspectives including the recovery success rate, average number of iterations and computational times. Comparison of the proposed algorithms with  a few existing ones is also made through the phase transition analysis including the phase transition cure and algorithm selection map. Simulations on random problem instances indicate that under proper choices of the algorithmic parameters, the algorithm HBHTP  is an efficient algorithm for sparse signal recovery and it may outperform several existing algorithms in many cases.

\end{document}